\documentclass[11pt, onecolumn,10pt, accepted=2020-09-04]{quantumarticle}

\usepackage[margin=1.5in]{geometry}
\usepackage[parfill]{parskip}
\usepackage[utf8]{inputenc}
\usepackage{physics}
\usepackage{bbm} 
\usepackage{color}
\usepackage{float}
\usepackage{hyperref} 
\usepackage{scrextend}
\usepackage{xargs}
\usepackage[dvipsnames]{xcolor}
\usepackage[T1]{fontenc}
\usepackage{xspace}
\usepackage[square,numbers,sort&compress]{natbib}

\usepackage{graphics}
\usepackage{graphicx}
\usepackage{pstricks}
\usepackage{pst-node}
\usepackage{pst-tree}
\usepackage{tikz} 
\usepackage{enumitem}
\setlistdepth{9}
\newlist{myEnumerate}{enumerate}{9}
\setlist[myEnumerate,1,2,3,4,5,6,7,8,9]{label*=\arabic*.} 
\renewlist{itemize}{itemize}{9}
\setlist[itemize]{label=\textbullet}
\setlist[itemize,2]{label=--}
\setlist[itemize,3]{label=*}
\setlist[itemize,5]{label=--}
\setlist[itemize,6]{label=*}
\setlist[itemize,8]{label=--}

\usepackage{amsmath,amssymb,amsfonts,amsthm}
\usepackage{blkarray} 
\usepackage{cleveref} 

\makeatletter
\DeclareRobustCommand{\bigDelta}{\mathop{\vphantom{\sum}\mathpalette\bigDelta@\relax}\slimits@}
\newcommand{\bigDelta@}[2]{\vcenter{\sbox\z@{#1\sum}\hbox{\resizebox{.9\dimexpr\ht\z@+\dp\z@}{!}{\m@th\Delta}}}}
\makeatother

\usepackage{xpatch}
\makeatletter
\AtBeginDocument{\xpatchcmd{\@thm}{\thm@headpunct{.}}{\thm@headpunct{}}{}{}}
\makeatother

\usepackage{algpseudocode} 
\usepackage{algorithm} 
\usepackage{subcaption}
\usepackage[colorinlistoftodos,prependcaption]{todonotes}

\newtheorem{theorem}{Theorem}[section]
\newenvironment{thm}{$\vspace{-0.5em}$\begin{theorem}}{\hfill$\diamond$\end{theorem}}
\crefname{thm}{theorem}{theorems}

\newtheorem{corollary}[theorem]{Corollary}
\newenvironment{cor}{$\vspace{-0.5em}$\begin{corollary}}{\hfill$\diamond$\end{corollary}}
\crefname{cor}{corollary}{corollaries}

\newtheorem{lemma}[theorem]{Lemma}
\newenvironment{lem}{$\vspace{-0.5em}$\begin{lemma}}{\hfill$\diamond$\end{lemma}}
\crefname{lem}{lemma}{lemmas}

\newtheorem{definition}[theorem]{Definition}
\newenvironment{mydef}{$\vspace{-0.5em}$\begin{definition}}{\hfill$\diamond$\end{definition}}
\crefname{mydef}{definition}{definitions}
\crefname{definition}{definition}{definitions}

\newtheorem{problem}[theorem]{Problem}
\newenvironment{pbm}{$\vspace{-0.5em}$\begin{problem}}{\hfill$\diamond$\end{problem}}
\crefname{pbm}{problem}{problems}
\crefname{problem}{problem}{problems}

\newcommandx{\unsure}[2][1=]{\todo[linecolor=red,backgroundcolor=red!25,bordercolor=red,#1]{#2}}
\newcommandx{\change}[2][1=]{\todo[linecolor=blue,backgroundcolor=blue!25,bordercolor=blue,#1]{#2}}
\newcommandx{\info}[2][1=]{\todo[linecolor=OliveGreen,backgroundcolor=OliveGreen!25,bordercolor=OliveGreen,#1]{#2}}
\newcommandx{\improvement}[2][1=]{\todo[linecolor=Plum,backgroundcolor=Plum!25,bordercolor=Plum,#1]{#2}}
\newcommandx{\thiswillnotshow}[2][1=]{\todo[disable,#1]{#2}}

\newcommand{\DefFormat}[1]{\texttt{\textsc{#1}}}
\newcommand{\NP}{\ensuremath{\mathbb{NP}}\xspace}
\newcommand{\Poly}{\ensuremath{\mathbb{P}}\xspace}
\newcommand{\EDP}{\DefFormat{EDP}\xspace}
\newcommand{\EDPDT}{\DefFormat{EDPDT}\xspace}
\newcommand{\BellVM}{\DefFormat{BellVM}\xspace}
\newcommand{\BellQM}{\DefFormat{BellQM}\xspace}
\newcommand{\LCLPMCC}{\ensuremath{\mathrm{LC}+\mathrm{LPM}+\mathrm{CC}}\xspace}
\newcommand{\bs}[1]{\ensuremath{\mathbf{#1}}}

\begin{document}
\title{Transforming graph states to Bell-pairs is NP-Complete}
\author{Axel Dahlberg Jonas Helsen Stephanie Wehner}
\date{QuTech - TU Delft, Lorentzweg 1, 2628CJ Delft, The Netherlands\\[2ex]
\today }

\maketitle
\begin{abstract}
    Critical to the construction of large scale quantum networks, i.e.
a quantum internet, is the development of fast algorithms for managing entanglement present in the network.
    One fundamental building block for a quantum internet is the distribution of Bell pairs between distant nodes in the network.
    Here we focus on the problem of transforming multipartite entangled states into the tensor product of bipartite Bell pairs between specific nodes using only a certain class of local operations and classical communication.
    In particular we study the problem of deciding whether a given graph state, and in general a stabilizer state, can be transformed into a set of Bell pairs on specific vertices using only single-qubit Clifford operations, single-qubit Pauli measurements and classical communication.
    We prove that this problem is \NP-Complete.
\end{abstract}


\section{Introduction}
Entanglement takes center stage in the modern understanding of quantum mechanics.
Apart from its usefulness as a theoretical tool, entanglement can also be seen as a resource that can be harnessed for secure communication and many other tasks, see e.g. ~\cite{Wehner2018stages}, not achievable by any protocol using only classical communication.
One can imagine a network of quantum-enabled nodes, a quantum internet, generating entanglement and harnessing it to perform tasks.
However, entanglement over large distances is very difficult to produce and even with rapidly improving technology, the amount of entanglement available in a network will for the foreseeable future be the limiting factor when performing the tasks mentioned above.
Moreover entanglement comes in different classes that are not necessarily mutually inter-convertible by operations performed locally on the nodes.
These two considerations: (1) the scarcity of entanglement as a resource and (2) the lack of local inter-convertibility of classes of entanglement sets the stage for the present work.
In this paper we assume that we already have some existing shared entangled state in a quantum network and we ask the question of whether this state can be transformed into a set of Bell pairs between specific nodes, using only a restricted set of local operations.
Examples of such a situation can be found in~\cite{pant2019routing,pant2019perculation}, where an approach is presented of first probabilistically generating a large graph state and then transforming this to the desired target state using local operations.
In~\cite{pant2019routing} this target state is precisely a set of Bell pairs between multiple pairs of nodes.

Any decision on this transformation process must be made fast, since entanglement decays over time~\cite{NielsenChuang}.
Hence there is a need for fast algorithms to decide whether different entangled states can be converted into each other by local operations.
In~\cite{vanmeter2018butterfly}, measurement-based quantum network coding was introduced, where one step in the procedure includes transforming general graph states into Bell-pairs using single-qubit Clifford operations, single-qubit Pauli measurement and classical communication: \LCLPMCC.
However, the computational complexity of finding the correct operations or even deciding if it can be done, was never mentioned.
In this work we answer this question.
Specifically we study the mathematic problem of deciding whether a given 'graph state' can be decomposed into a series of Bell pairs on specific vertices, using only \LCLPMCC.
In particular we determine the computational complexity of this problem.
We consider this set of operations since they are on many hardware platforms the simplest and fastest operations to perform.
For example on Nitrogen-Vacancies in diamond, single-qubit gates on one of the qubits (electron) is many orders faster than two-qubits gates and single-qubit gates on the other qubits (nuclear spins)~\cite{Kalb2017}.
Furthermore, restricting to only Cliffords, instead of allowing general unitary operations, is natural on systems where, for example, the qubits are logical qubits of some error-correcting code.
Since for many quantum error-correction codes, doing non-Clifford operations is a costly process involving, in many cases, the consumption of magic states.
While it is true that the high relative cost of logical non-Clifford gates is not a universal phenomena, for instance 3D color codes~\cite{bombin2007exact,bombin2015gauge} have the ability to perform transversal T-gates, it is the case for many popular codes.
Moreover, over the years many techniques from graph theory have been put to bear on specifically the behavior of graph states under \LCLPMCC~\cite{hein2006survey,VandenNest2004graphical,dahlberg2018courcelle,dahlberg2018long,hahn2018quantum,dahlberg2019iso}.
This gives us the ability to ask very precise questions and also answer them.

Our main result is that we prove that the problem of deciding whether a given graph state can be converted into Bell Pairs using only \LCLPMCC\ (\BellVM) is in general \NP-Complete.
This means that unless $\Poly=\NP$, there will be no efficient algorithm for solving this problem on arbitrary graph states.
In order to prove our results we make heavy use of results in algorithmic graph theory, and in fact we prove new graph theoretical results in the process of proving our main theorem.

\section*{Related work}
The study of single-qubit operations on graph states started with the work by Raussendorf and Briegel in~\cite{Raussendorf2001persistent} and lead to the development of measurement-based quantum computing~\cite{Raussendorf2001oneway}.
In~\cite{VandenNest2004graphical} Van den Nest et al. showed how the action of single-qubit Clifford operations can be described purely by the graph operation \emph{local complementation} on the corresponding graph.
Bouchet had developed back in 1991 (~\cite{Bouchet1991efficient}) an efficient algorithm to decide if two graphs are equivalent under local complementations which was used by Van den Nest et al. in~\cite{VandenNest2004efficient} to present an efficient algorithm for deciding the equivalence of graph states under single-qubit Clifford operations.
Adding single-qubit Pauli measurements and classical communication to the set of operations is equivalent to also allow for vertex-deletion on the corresponding graph, next to the local complementations~\cite{dahlberg2018courcelle}.
Graphs that can be reached by the action of local complementations and vertex-deletions is a known concept in graph theory and are called vertex-minors~\cite{Oum2005rankwidth,Courcelle2007vertexminor}.
Using the fact that the graph states reachable from a graph state $\ket{G}$ using \LCLPMCC\ are exactly captured by the vertex-minors of $G$, we have previously shown that deciding if a GHZ-state can be reached from a given graph state using \LCLPMCC\ is \NP-Complete, both if (1) the GHZ-state is on a given set of vertices~\cite{dahlberg2018long} and (2) the GHZ-state is allowed to be on any subset of a fixed size~\cite{dahlberg2019iso}.

Due to the different \emph{target state} considered in this paper, the computational complexity of \BellVM does not follow from~\cite{dahlberg2018long}.
Concretely, a tensor-product of Bell pairs is not equivalent to a GHZ-state under the operations studied here.
For this reason we will here perform a different reduction to prove that \BellVM is \NP-complete compared to the one used to prove the same for the problem in~\cite{dahlberg2018long}.
In~\cite{dahlberg2018long} we reduced from the problem of finding a Hamiltonian tour on a three-regular graph.
Here on the other hand we reduce from the problem of finding disjoint paths between sources and sinks in a graph.
While doing so we also show new results, purely related to the disjoint path problem in graph theory.

However, the fact that a problem is \NP-Complete does not mean that there is no efficient algorithm, if one allows to put certain restrictions on the input to the problem.
Such a restriction is for example that a certain parameter $r$ of the input should be bounded.
If a problem can be solved in polynomial time on inputs where this parameter $r$ is bounded, the problem is said to be fixed-parameter tractable (in $r$).
Concretely, there might exist an algorithm solving the problem in time $\mathcal{O}(f(r)\cdot\text{poly}(n))$, where $f$ is a computable function and $n$ is the size of the input.
For \NP-Complete problems, $f(r)$ is necessarily super-polynomial in $n$, unless $\Poly=\NP$.
Individual problems have been shown to be fixed-parameter tractable in various parameters.
However, Courcelle~\cite{Courcelle2011book} showed that a large class of graph problems are fixed-parameter tractable in what is called the rank-width~\cite{Oum2005rankwidth} of the graph.
In fact, any graph problem, expressible in a certain rich logic (MS)\footnote{Monadic second-order logic}, can be solved in time $\mathcal{O}(f(\mathrm{rwd}(G))\cdot\abs{V(G)}^3)$, where $\mathrm{rwd}(G)$ is the rank-width of $G$ and $\abs{V(G)}$ is the number of vertices of $G$.
In~\cite{Courcelle2007vertexminor} Courcelle and Oum showed that the vertex-minor problem is fixed-parameter tractability in the rank-width of the input graph by showing that the problem is expressible in MS.
The rank-width of a graph $G$ equals one plus the \emph{Schmidt-rank width} the graph state $\ket{G}$~\cite{VandenNest2007schmidt}.
Using these results, we applied Courcelle's theorem to the problem of transforming graph states under $\mathrm{LC}+\mathrm{LPM}+\mathrm{CC}$ in~\cite{dahlberg2018courcelle} and thus showed that this problem is fixed-parameter tractable in the Schmidt-rank width of the input graph state.

Courcelle's theorem states that there exists an efficient algorithm for graph states with bounded Schmidt-rank width.
However, a direct implementation of the algorithm from Courcelle's theorem is not usable in practice, due to a huge constant factor in the runtime~\cite{Langer2014practical}.
In~\cite{dahlberg2018long} we presented two efficient algorithms, not suffering from this huge constant factor, for the problem of deciding whether $\ket{G}$ can be transformed to $\ket{G'}$ using \LCLPMCC\ if: (1) $\ket{G'}$ is a GHZ-state and $\ket{G}$ has Schmidt-rank width 1~\cite{dahlberg2018long} and (2) $\ket{G'}$ is a GHZ-state of bounded size and $G$ is a circle graph~\cite{dahlberg2018long}.
We point out that the second algorithm is in fact not captured by Courcelle's theorem since circle graphs have unbounded rank-width~\cite{dahlberg2018long}.

Here we show a hardness proof for the case when the target graph is the tensor product of some number of Bell pairs on specified qubits.
However, our result is not only a negative one, since it can also give a hint for heuristic algorithms that can be used to solve the problem at hand, which do not suffer from a huge constant pre-factor in the runtime as for those based on Courcelle's theorem.
In particular, since we show the relation between \BellVM and certain problems related to finding disjoint paths in a graph, we can make use of results for solving the disjoint path problem also for \BellVM.
Even though the disjoint path problem is in general \NP-complete, there are many heuristic algorithms.
For example in~\cite{martin2020edp_heuristic}, a two-step approach using an integer linear program formulation of the problem with an evolutionary algorithm.
By using the reduction we show in this paper, one might be able to use the same approach for the problem at hand.

\section*{Overview}
In \cref{sec:graph_states} we review the basic theory of graph states and introduce the notion of a Bell-vertex-minor, a graph theoretical concept central to our results.
In \cref{sec:circle} we review the notion of circle graphs and related concepts which we make use of in \cref{sec:bellvmnpc}.
In \cref{sec:problems} we discuss the Edge-Disjoint path problem, a well-studied computational problem in algorithmic graph theory.
In \cref{sec:bellvmnpc} we show that a certain version of the Edge Disjoint Path problem can be polynomially reduced to \BellVM\ and in \cref{sec:edpnpc} that this version of the Edge Disjoint Path problem is \NP-Complete, implying that \BellVM\ is \NP-hard.
Finally, in \cref{sec:conclusion} we discuss the implications of our result.

\subsection*{Notation}
Graphs are assumed to be simple unless otherwise indicated. The vertex-set of a graph $G=(V,E)$ is denoted $V=V(G)$ and the edge-set is denoted $E=E(G)$. Given a vertex $v$ in a graph $G$ we denote the neighborhood of $v$ (the set of vertices adjacent to $v$ in $G$) by $N_v^{(G)}$. If it is clear which graph the neighborhood concerns we sometimes omit $G$ and simply write $N_v$. Given a graph $G$ and a subset of its vertices $V'$ we will denote the induced subgraph of $G$ on those vertices by $G[V']$.
We denote the fully connected graph on $n$ vertices as $K_n$.

Throughout this paper we use the following notation for sets of consecutive natural numbers
\begin{align}
    [k, n] &\equiv \{i\in\mathbb{N}: k\leq i\leq n\} \\
    [n] &\equiv \{i\in\mathbb{N}: 1\leq i\leq n\}
\end{align}

\subsection{Graph states}\label{sec:graph_states}
A graph state is a multi-partite quantum state $\ket{G}$ which is described by a graph $G$, where the vertices of $G$ correspond to the qubits of $\ket{G}$.
The graph state is formed by initializing each qubit $v\in V(G)$ in the state $\ket{+}_v=\frac{1}{\sqrt{2}}(\ket{0}_v+\ket{1}_v)$ and for each edge $(u,v)\in E(G)$ applying a controlled phase gate between qubits $u$ and $v$.
Importantly, all the controlled phase gates commute and are invariant under changing the control- and target-qubits of the gate.
This allows the edges describing these gates to be unordered and undirected.
Formally, a graph state $\ket{G}$ is given as
\begin{equation}
    \ket{G}=\prod_{(u,v)\in E(G)}C_Z^{(u,v)}\left(\bigotimes_{v\in V(G)}\ket{+}_v\right),
\end{equation}
where $C_Z^{(u,v)}$ is a controlled phase gate between qubit $u$ and $v$, i.e.
\begin{equation}
    C_Z^{(u,v)}=\ket{0}\bra{0}_u\otimes\mathbb{I}_v+\ket{1}\bra{1}_u\otimes Z_v
\end{equation}
and $Z_v$ is the Pauli-$Z$ matrix acting on qubit $v$.



One advantage of considering graph states is that certain problems can be completely expressed in the language of graph theory, where one can make use of powerful existing tools and techniques.
For example, what we make use of here, is the fact that single-qubit Clifford operations (LC), single-qubit Pauli measurements (LPM) and classical communication (CC): \LCLPMCC, which take graph states to graph states, can be completely characterized by local complementations and vertex-deletions on the corresponding graphs~\cite{dahlberg2018courcelle}.
Local complementation is defined as follows.

\begin{mydef}[Local complementation]\label{def:LC}
    A local complementation $\tau_v$ is a graph operation specified by a vertex $v$, taking a graph $G$ to $\tau_v(G)$ by replacing the induced subgraph on the neighborhood of $v$, i.e.  $G[N_v]$, by its complement.
    The neighborhood of any vertex $u$ in the graph $\tau_v(G)$ is therefore given by
    \begin{equation}
    N_u^{(\tau_v(G))}=\begin{cases}N_u\Delta (N_v\setminus\{u\}) & \quad \text{if } (u,v)\in E(G) \\ N_u & \quad \text{else}\end{cases},
    \end{equation}
    where $\Delta$ denotes the symmetric difference between two sets.
    Given a sequence of vertices $\bs{v}=v_1\dots v_k$, we denote the induced sequence of local complementations, acting on a graph $G$, as
    \begin{equation}
        \tau_{\bs{v}}(G)=\tau_{v_k}\circ\dots\circ\tau_{v_1}(G).
    \end{equation}
\end{mydef}

The action of a local complementation on a graph induce the following sequence of single-qubit Clifford operations on the corresponding graph state
\begin{equation}\label{eq:LCU}
	U_v^{(G)}=\exp\left(-\mathrm{i}\frac{\pi}{4}X_v\right)\prod_{u\in N_v}\exp\left(\mathrm{i}\frac{\pi}{4}Z_u\right),
\end{equation}
where $X_v$ and $Z_v$ are the Pauli-X and Pauli-Z matrices acting on qubit $v$ respectively.
Concretely, $U_v^{(G)}$ has the following action on the graph state $\ket{G}$
\begin{equation}
    U_v^{(G)}\ket{G}=\ket{\tau_v(G)}.
\end{equation}

Measuring qubit $v$ of a graph state $\ket{G}$ in the Pauli-$X$ (or Pauli-$Y$, Pauli-$Z$ basis), gives a stabilizer state that is single-qubit Clifford equivalent to a graph state $\ket{G'}$, where $G'$ can be reached from $G$ by a sequence of local complementations and vertex-deletions~\cite{hein2006survey}. 
The state operations taking the post-measurement state to $\ket{G'}$ depends on the measurement outcome and acts on $v$ and its neighborhood.
This means classical communication is required to announce the measurement result at the vertex $v$ to its neighboring vertices.
Details can be found in~\cite{dahlberg2018long}, where we introduced to notion of a \emph{qubit-minor} which captures exactly which graph states can be reached from some initial graph state under  \LCLPMCC.
Formally we define a qubit-minor as:

\begin{mydef}[Qubit-minor~\cite{dahlberg2018long}]\label{def:QM} Assume $\ket{G}$ and $\ket{G'}$ are graph states on the sets of qubits $V$ and $U$ respectively.
    $\ket{G'}$ is called a qubit-minor of $\ket{G}$ if there exists an adaptive sequence of single-qubit Clifford operations (LC), single-qubit Pauli measurements (LPM) and classical communication (CC) that takes $\ket{G}$ to $\ket{G'}$, i.e.
    \begin{equation}
        \ket{G}\xrightarrow[\mathrm{LPM}+\mathrm{CC}]{\mathrm{LC}}\ket{G'}\otimes\ket{\mathrm{junk}}_{V\setminus U}.
    \end{equation}
    If $\ket{G'}$ is a qubit-minor of $\ket{G}$, we denote this as
    \begin{equation}
        \ket{G'}<\ket{G}.
    \end{equation}
\end{mydef}

In~\cite{dahlberg2018long} we have shown that the notion of qubit-minors for graph states is equivalent to the notion of \emph{vertex-minors} for graphs.

\begin{mydef}[Vertex-minor]\label{def:vertex-minor}
	A graph $G'$ is called a vertex-minor of $G$ if and only if there exist a sequence of local complementations and vertex-deletions that takes $G$ to $G'$~\cite{Oum2005rankwidth,Courcelle2007vertexminor}.
    Since vertex-deletions can always be performed last in such a sequence (see~\cite{dahlberg2018courcelle}), an equivalent definition is the following: A graph $G'$ is called a vertex-minor of $G$ if and only if there exist a sequence of local complementations $\bs{v}$ such that $\tau_{\bs{v}}(G)[V(G')]=G'$.
    If $G'$ is a vertex-minor of $G$ we write this as
    \begin{equation}
        G'<G.
    \end{equation}
\end{mydef}

The relation between qubit-minor and vertex-minor is captured by the following theorem.

\begin{thm}[\cite{dahlberg2018long}]\label{thm:QMVM}
    Let $\ket{G}$ and $\ket{G'}$ be two graph states such that no vertex in $G'$ has degree zero.
    Then $\ket{G'}$ is a qubit-minor of $\ket{G}$ if and only if $G'$ is a vertex-minor of $G$, i.e.
    \begin{equation}
        \ket{G'}<\ket{G}\quad\Leftrightarrow\quad G'<G.
    \end{equation}
\end{thm}

Note that one can also include the case where $G'$ has vertices of degree zero.
Let's denote the vertices of $G'$ which have degree zero as $I$.
We then have that
\begin{equation}
    \ket{G'}<\ket{G}\quad\Leftrightarrow\quad G'[V(G')\setminus I]<G.
\end{equation}

\Cref{thm:QMVM} is very powerful since it allows us to consider graph states under \LCLPMCC, purely in terms of vertex-minors of graphs.
We will therefore in the rest of this paper use the formalism of vertex-minors to study the computational complexity of transforming graph states using \LCLPMCC.

\subsection{Bell vertex-minors}
In this paper we are interested in the question of whether a given graph state $\ket{G}$ can be transformed into some number of Bell pairs using \LCLPMCC\ between specific vertices.
We will denote the following Bell pair on qubits $a$ and $b$ as
\begin{equation}
    \ket{\Phi^+}_{ab} = \frac{1}{\sqrt{2}}(\ket{0}_a\otimes\ket{0}_b+\ket{1}_a\otimes\ket{1}_b).
\end{equation}
We formally define the following main problem of this paper.

\begin{pbm}[\BellQM]\label{pbm:bellqm}
	Given a graph state $\ket{G}$ and a set of disjoint pairs $B=\{\{p_1,p_1'\},\dots,\{p_k,p_k'\}\}$.
	Let $\ket{G_B}$ be the state following state consisting of Bell pairs between each pair of $B$
	\begin{equation}
		\ket{G_B}=\bigotimes_{\{p,p'\}\in B}\ket{\Phi^+}_{pp'}
	\end{equation}
	Decide if $\ket{G_B}$ is a qubit-minor of $G$.
\end{pbm}

The graph state described by the complete graph on two vertices $K_2$ is single-qubit Clifford equivalent to each of the four Bell pairs since
\begin{equation}
    \ket{K_2} = \frac{1}{\sqrt{2}}(\ket{0}_a\otimes\ket{+}_b+\ket{1}_a\otimes\ket{-}_b) = H_b\ket{\Phi^+}_{ab}
\end{equation}
where $\ket{+}=(\ket{0}+\ket{1})/\sqrt{2}$, $H_b$ is a Hadamard gate on qubit $b$.

Using \cref{thm:QMVM} we can turn the question of transforming graph states to Bell pairs using \LCLPMCC\, i.e. \BellQM, into the question of whether a disjoint union of $K_2$'s is a vertex-minor of some graph.
Formally we have the following problem graph problem.

\begin{pbm}[\BellVM]\label{pbm:bellvm}
    Given a graph $G$ and a set of disjoint pairs $B=\{\{p_1,p_1'\},\dots,\{p_k,p_k'\}\}$.
    Let $G_B$ be the graph with vertices $V(G_B)=\bigcup_{i\in[k]}\{p_i,p_i'\}$ and edges $E(G_B)=\bigcup_{i\in[k]}\{(p_i,p_i')\}$.
    Decide if $G_B$ is a vertex-minor of $G$.
\end{pbm}

To be precise $(\ket{G},B)$ is a YES-instance of \BellQM, if and only if $(G,B)$ is a YES-instance of \BellVM.

\begin{figure}[H]
    \centering
    \begin{subfigure}{0.2\textwidth}
        \includegraphics[width=1\textwidth]{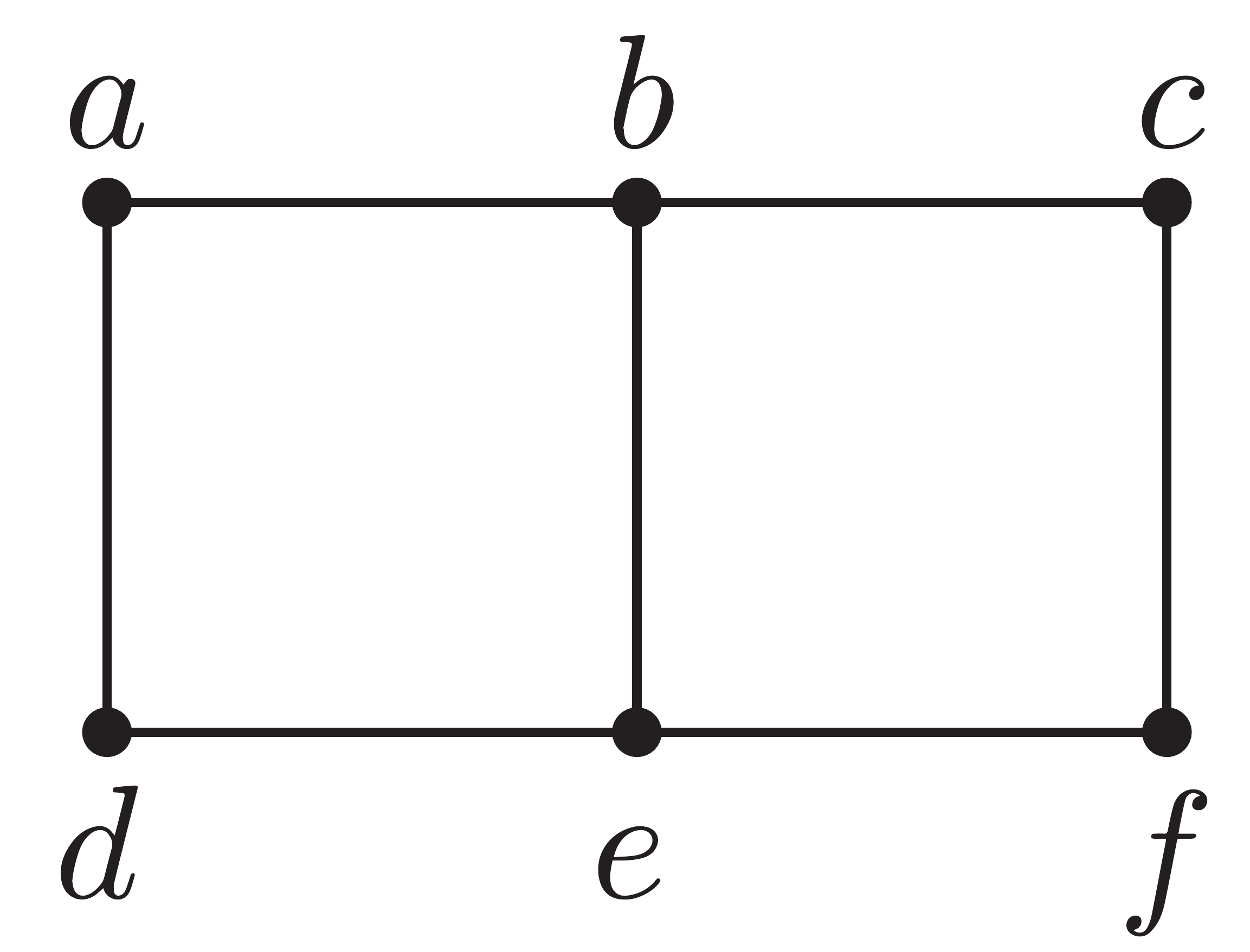}
        \caption{}
        \label{fig:bell_example_1}
    \end{subfigure}
    \raisebox{2ex}{$\quad\xrightarrow{\tau_b\circ\tau_e\circ\tau_b}\quad$}
    \begin{subfigure}{0.2\textwidth}
        \includegraphics[width=1\textwidth]{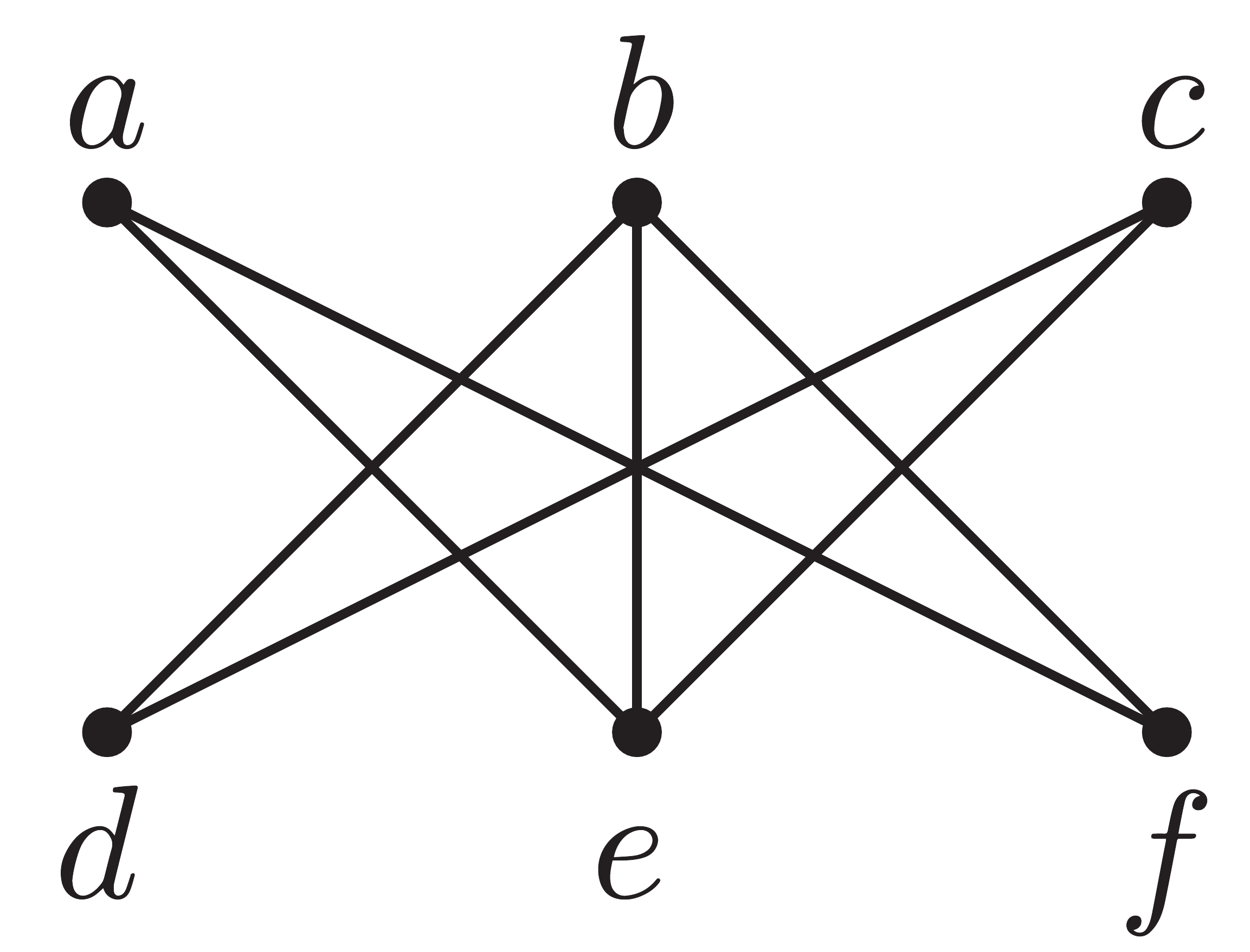}
        \caption{}
        \label{fig:bell_example_2}
    \end{subfigure}
    \raisebox{2ex}{$\quad\xrightarrow{\setminus\{b,e\}}\quad$}
    \begin{subfigure}{0.2\textwidth}
        \includegraphics[width=1\textwidth]{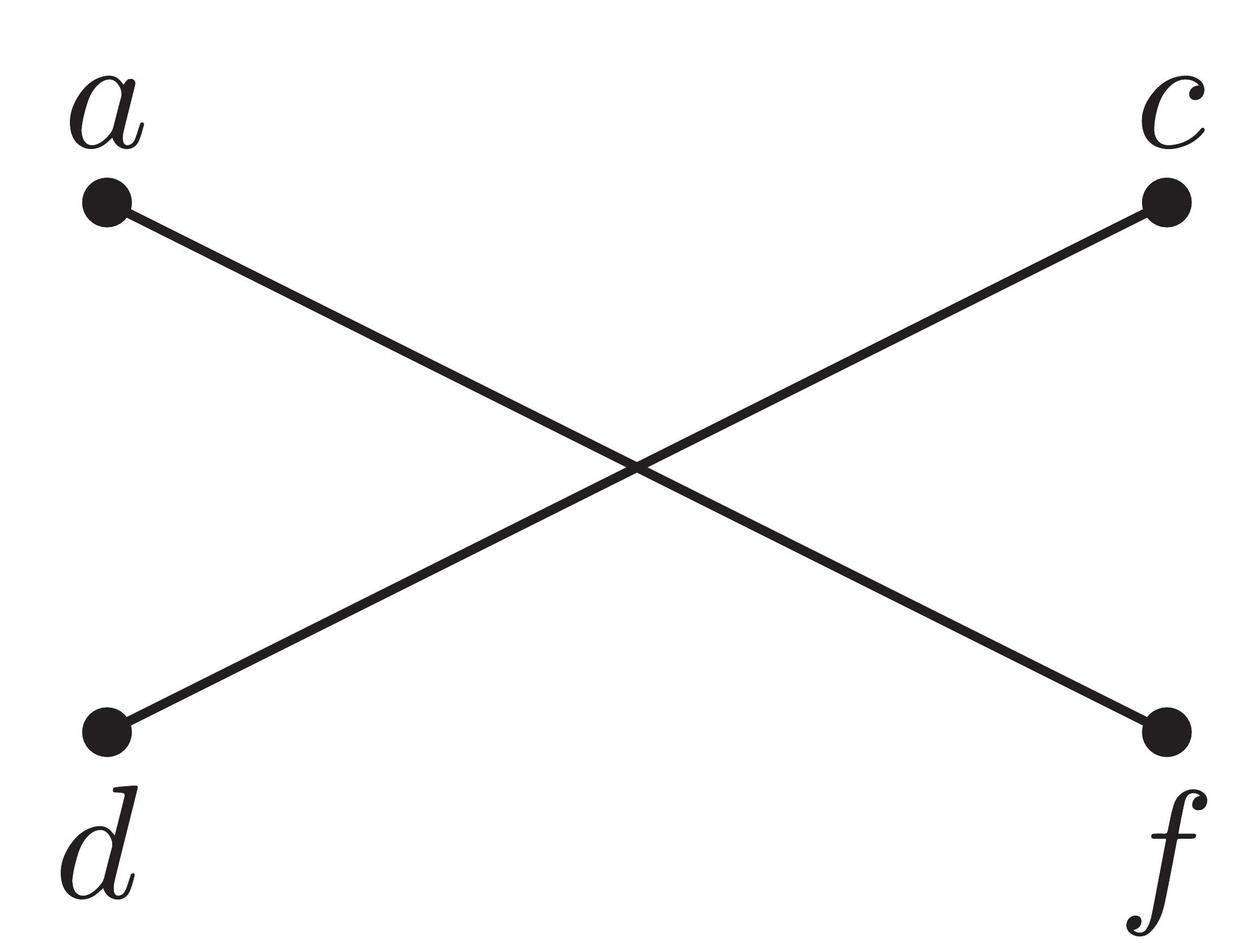}
        \caption{}
        \label{fig:bell_example_3}
    \end{subfigure}
    \caption{An example of a graph being transformed to a union of two $K_2$ graphs, using local complementations and vertex-deletions.
    The original graph in \cref{fig:bell_example_1} get transformed to the graph in \cref{fig:bell_example_2} by performing local complementations on the vertices $b$, $e$ and then $b$ again.
    Finally, when vertices $b$ and $e$ are deleted, the disjoint union of two $K_2$ graphs, on vertices $\{a,f\}$ and $\{c,d\}$ respectively, is reached (\cref{fig:bell_example_3}).}
    \label{fig:bell_example}
\end{figure}

\section{Circle graphs}\label{sec:circle}

Circle graphs are the graphs which can be described as a Eulerian tours on 4-regular multi-graphs, as described below.
What is interesting in this context is that circle graphs are equivalent under local complementations if and only if they can be described as Eulerian tours on the same 4-regular multi-graph~\cite{Bouchet1988graphicisotropic}.
We will make use of this property to prove that \BellVM\ is \NP-Complete, and therefore also \BellQM, in \cref{sec:bellvmnpc}.
Here we review circle graphs and certain properties we will later need.
For more details on circle graphs, see for example ~\cite{Kotzig1977,Bouchet1972,Bouchet1994circle}, the book by Golumbic~\cite{Golumbic2004} or ~\cite{dahlberg2018long} for the use of circle graphs in the context of graph states.


A $4$-regular multi-graph is a graph where each vertex has exactly four incident edges and can contain multiple edges between each pair of vertices or edges only incident to a single vertex (self-loops).

\begin{mydef}[Eulerian tour]
	Let $F$ be a connected multi-graph.
	A \emph{walk} $W$ on $F$ is an alternating sequence of vertices and edges
	\begin{equation}
		W = v_1e_1v_2\dots e_{k}v_{k+1}
	\end{equation}
	such that $e_i$ is incident on $v_i$ and $v_{i+1}$ for $i\in[n]$.
	A tour on $F$ is a closed walk, i.e. $v_1=v_{k+1}$, without repeated edges.
	An Eulerian tour $U$ on $F$ is a tour that visits each edge in $F$ exactly once.
\end{mydef}

Any $4$-regular multi-graph is Eulerian, i.e. has a Eulerian tour, since each vertex has even degree~\cite{biggs1976graph}.

Furthermore, any Eulerian tour on a $4$-regular multi-graph $F$ traverses each vertex exactly twice, except for the vertex which is both the start and the end of the tour.
The order in which these vertices are traversed is captured by the \emph{induced double-occurrence word}.

\begin{mydef}[Induced double-occurrence word]\label{def:eul_tour}
    Let $F$ be a connected $4$-regular multi-graph on $k$ vertices $V(F)$.
    Let $U$ be a Eulerian tour on $F$ of the form
    \begin{equation}\label{eq:eul_tour}
        U=x_1e_1x_2\dots x_{2k-1}e_{2k-1}x_{2k}e_{2k}x_1.
    \end{equation}
    with $x_i\in V(F)$ and $e_i\in E(F)$.
    From a Eulerian tour $U$ as in \cref{eq:eul_tour} we define an induced double-occurrence word as
    \begin{equation}
        m(U)=x_1x_2\dots x_{2k-1}x_{2k}.
    \end{equation}
\end{mydef}

We will now define a mapping from an induced double-occurrence word $m(U)$ to a graph $\mathcal{A}(m(U))$, where the edges of $\mathcal{A}(m(U))$ are exactly the pairs of vertices in $m(U)$ which alternate.
Formally we have the following definition.

\begin{mydef}[Alternance graph]\label{def:alt_graph}
    Let $m(U)$ be the induced double-occurrence word of some Eulerian tour $U$ on some $4$-regular multi-graph $F$.
    Let now $\mathcal{A}(m(U))$ be a graph with vertices $V(F)$ and the edges $E$, such that for all $(u,v)\in V(F)\otimes V(F)$, $(u,v)\in E$ if and only if $m(U)$ is of the form
    \begin{equation}
        \dots u \dots v \dots u \dots v \dots \quad\text{or}\quad \dots v \dots u \dots v \dots u \dots,
    \end{equation}
    i.e. $u$ and $v$ are alternating in $m(U)$.
    We will sometimes also write $\mathcal{A}(U)$ as short for $\mathcal{A}(m(U))$.
\end{mydef}

\Cref{fig:example} shows an example of a 4-regular multi-graph and one of its induced alternance graphs.

\begin{figure}[H]
    \centering
    \begin{subfigure}{0.35\textwidth}
        \includegraphics[width=1\textwidth]{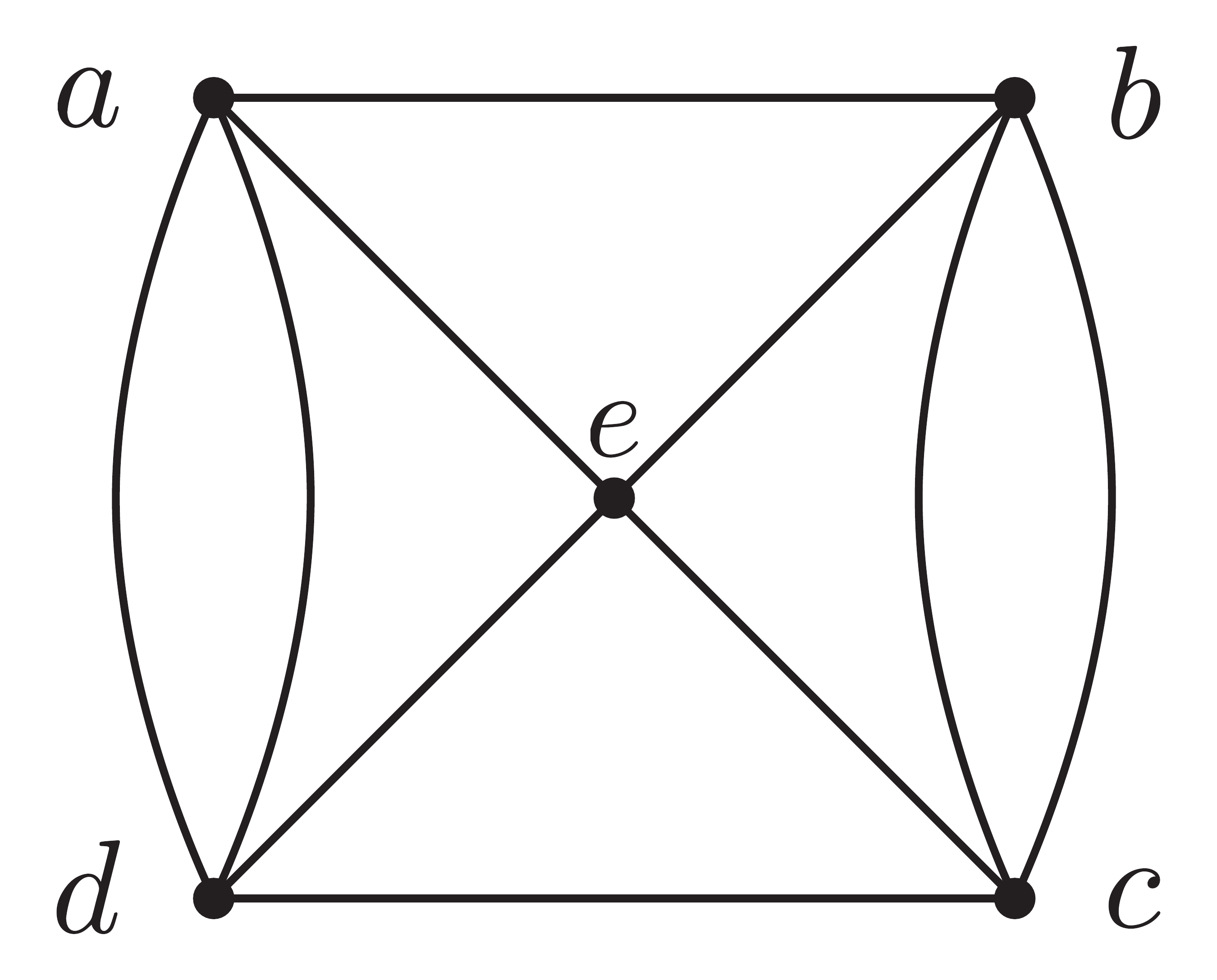}
        \caption{}
        \label{fig:example_a}
    \end{subfigure}
    \hspace{2cm}
    \begin{subfigure}{0.35\textwidth}
        \includegraphics[width=1\textwidth]{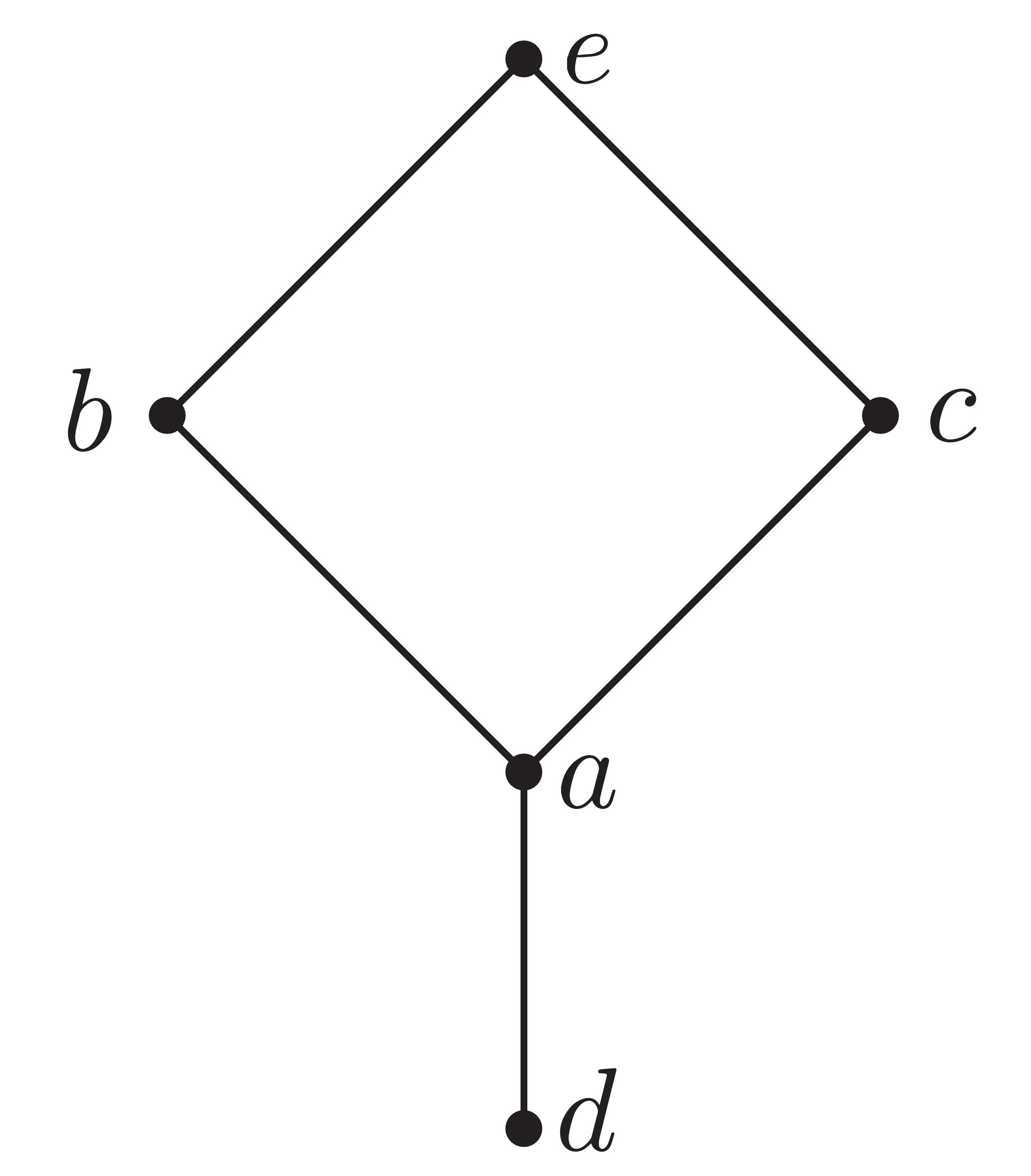}
        \caption{}
        \label{fig:example_b}
    \end{subfigure}
    \caption{
        \Cref{fig:example_a} shows an example of a $4$-regular multi-graph $F$, which has an Eulerian tour $U_0$ with an induced double-occurrence word $m(U) = adcbaebced$.
        The alternating vertex-pairs of $m(U_0)$ are thus $(a,b),\; (a,c),\; (a,d),\; (b,e),\; (c, e)$ and their mirrors.
        The alternance graph $\mathcal{A}(m(U_0))$ is therefore the graph in \cref{fig:example_b}.
    }
    \label{fig:example}
\end{figure}

Circle graphs are exactly the graphs which are the alternance graph described by some Eulerian tour on some $4$-regular multi-graph~\cite{Bouchet1994circle}.
Given a Eulerian tour $U$ of some $4$-regular multi-graph $F$ and a subset $V'$ of $V(F)$ we will sometimes write $m(U)[V']$ to mean the word $m(U)$ with all letters in $V(F)\setminus V'$ deleted.
The notation is intentionally similar to that of taking an induced subgraph since the alternance graph of $m(U)[V']$ is exactly the induced subgraph of the alternance graph of $m(U)$, i.e.
\begin{equation}\label{eq:subword}
    \mathcal{A}\Big(m(U)[V']\Big) = \mathcal{A}\Big(m(U)\Big)[V'].
\end{equation}

Importantly here is that one can answer whether a circle graph has a certain vertex-minor by considering the Eulerian tours of a certain $4$-regular multi-graph.
Formally we have the following theorem, which is proven in~\cite{dahlberg2018long}.

\begin{thm}[\cite{dahlberg2018long}]\label{thm:vm_of_eul}
    Let $F$ be a connected $4$-regular multi-graph and let $G$ be a circle graph such that $G=\mathcal{A}(U)$ for some Eulerian tour $U$ on $F$.
    Then $G'$ is a vertex-minor of $G$ if and only if there exist a Eulerian tour $U'$ on $F$ such that
    \begin{equation}\label{eq:vm_of_eul}
        G'=\mathcal{A}\Big(m(U')\big[V(G')\big]\Big).
    \end{equation}
\end{thm}

Using \cref{thm:vm_of_eul}, we can now ask what property a 4-regular multi-graph should have, such that it's induced alternance graphs are YES-instances to \BellVM, given a set of disjoint pairs $P$.
As we show in \cref{lem:bellEDP_equiv}, this question will be directly related to a restricted version of the edge-disjoint path problem, which we define in the next section.

\section{The Edge-disjoint path problem}\label{sec:problems}
In \cref{sec:bellvmnpc} we show that \BellVM\ is \NP-Complete by reducing the $4$-regular \EDPDT\ (Edge Disjoint Paths with Disjoint Terminals) problem (see below) to \BellVM.
We then show that that $4$-regular \EDPDT\ is \NP-Complete, see \cref{cor:edpnpc}, which therefore implies that the same is true for \BellVM.
This is done by reducing the EDP (Edge Disjoint Path) problem to \EDPDT\ and then \EDPDT\ to $4$-regular \EDPDT.
We begin by formally defining all problems just mentioned.
We will denote the set of edges in a path $P=v_0e_1v_1\dots e_lv_l$ as $E(P)=\{e_1,\dots,e_l\}$.
Moreover, given two graphs $G=(V(G), E(G))$ and $D=(V(D), E(D))$ we denote by $G\cup D$ the graph formed by the vertices $V(G\cup D)=V(G)\cup V(D)$ and the edge multi-set $E(G\cup D) = E(G)\cup E(D)$.
A path on a graph is a walk without repeated vertices.
A closed path is called a circuit.

We first define the EDP (Edge Disjoint Path) problem.

\begin{pbm}[EDP]\label{pbm:edp}
Let $G$ and $D$ be graphs such that $V(D)\subseteq V(G)$.
Decide whether there exists a set of edge-disjoint circuits $\mathcal{C}$ on the graph $G\cup D$ such that every edge $e\in E(D)$ is part of exactly one circuit in $\mathcal{C}$.
\end{pbm}

\begin{figure}[H]
    \centering
    \includegraphics[width=0.5\textwidth]{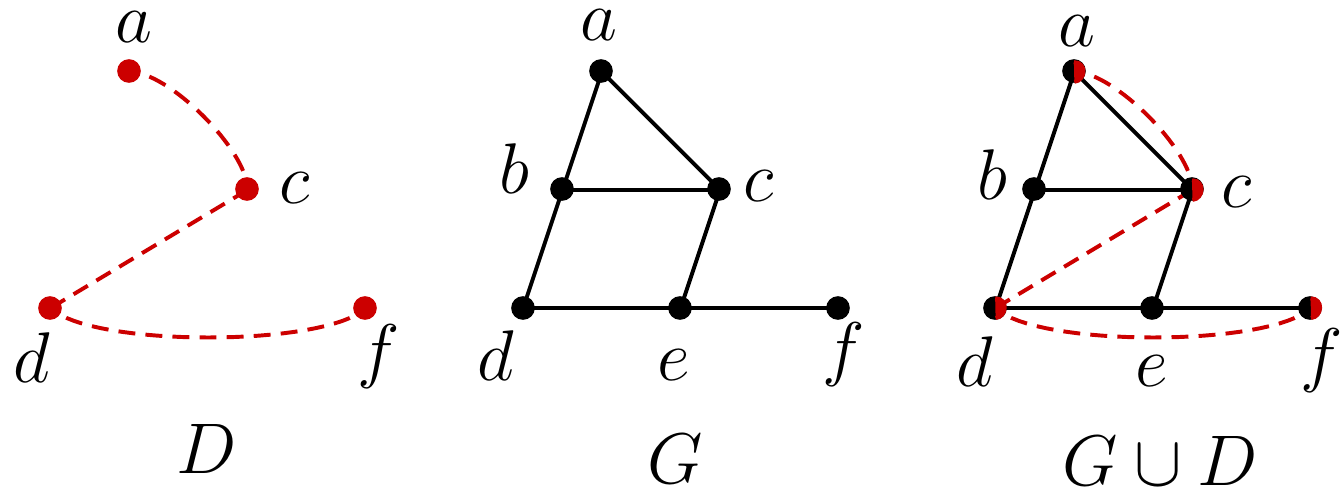}
    \caption{The EDP problem concerns deciding if there exists a set of edge-disjoint circuits in $G\cup D$ such that each edge in $D$ is in part of exactly one of the circuits.
        In this example a satisfying solution is a set consisting of the three circuits $a(a,c)c{\color{red}(c,a)}a$, $d(d,b)b(b,c)c{\color{red}(c,d)}d$ and $d(d,e)e(e,f)f{\color{red}(f,d)}d$.
    }
    \label{fig:EDP}
\end{figure}

Let $(G,D)$ be a YES-instance of \EDP\ and let $\mathcal{C}$ be the edge-disjoint circuits of \cref{pbm:edp}.
For each edge $e$ in $E(D)$ denote the circuit in $\mathcal{C}$ which $e$ is part of as $C_e$.
These edge disjoint circuits $C_e$ on $G\cup H$ correspond to edge disjoint paths on the graph $G$ with \emph{terminals}, i.e. beginning and ending points of the path, precisely at the vertices $u,v$ s.t.
$e = (u,v)$ for all $e \in D$.
The EDP problem is known to be \NP-Complete even in the case where $G\cup D$ is a Eulerian graph~\cite{vygen1995np}.
The EDPDT problem is now the EDP problem but with the demand graph $D$ restricted to be a disjoint union of connected graphs on two vertices, i.e.
of the form $K_2^{\times k}$ for some $k$, such that the terminals are distinct.

\begin{pbm}[\EDPDT]
Let $G$ and $D=K^{\times k}_2$ be graphs such that $V(D)\subseteq V(G)$.
Decide \EDP\ with the instance $(G, D)$.
\end{pbm}

The 4-regular \EDPDT\ problem is then a further restriction of this problem to the case where $G\cup D$ is $4$-regular (and $D = K^{\times k}_2$).
Formally we have

\begin{pbm}[4-regular \EDPDT]\label{pbm:4_regEDPDT}
Let $G$ and $D=K^{\times k}_2$ be graphs such that $V(D)\subseteq V(G)$ and $G\cup D$ is $4$-regular.
Decide \EDP\ with the instance $(G, D)$.
\end{pbm}

Note that the $4$-regularity of $G\cup D$ means that all vertices in $G$ which are not in $D$ must have degree $4$, while all vertices also in $D$ must have degree $3$.
We can equivalently formulate the $4$-regular EDPDT problem as a problem involving only $G$ and a set of terminal pairs on $G$, which will be a more useful definition for some of the proofs.

\begin{pbm}[4-regular EDPDT (equivalent formulation)]\label{pbm:4_regEDPDTequiv}
    Let $G$ be a multi-graph where each vertex has degree either 3 or 4.
    Let $T=\{\{t_1,t'_1\},\dots,\{t_k,t_k'\}\}$ be the set of disjoint terminal pairs, such that $t_i,t_i'\in V(G)$ and $d_G(t_i)=d_G(t_i')=3$ for all $i\in[k]$.
    Furthermore, assume that there are no other vertices of degree $3$ in $G$, i.e. $\bigcup_{i\in[k]}\{t_i,t_i'\}=\{v\in V(G):d_G(v)=3\}$. 
    Decide if there exists $k$ edge-disjoint paths $P_i$ for $i\in[k]$ such that the ends of $P_i$ are $t_i$ and $t_i'$.
\end{pbm}

In the appendix (\cref{thm:EDP2EDPDT}) we show that \EDP\ can be reduced to $4$-regular \EDPDT\ in polynomial time.
A corollary to this theorem is therefore.

\begin{cor}\label{cor:edpnpc}
    4-regular \EDPDT\ is \NP-Complete.
\end{cor}

\section{BellVM is NP-Complete}\label{sec:bellvmnpc}

Now we move on to proving the main result of this paper.
The 4-regular \EDPDT\ problem can be reduced to the \BellVM\ problem as shown below in \cref{thm:EDPDT2BellVM}.
Since we also show that \BellVM\ is in \NP we conclude that \BellVM\ is \NP-Complete.

\begin{cor}\label{cor:bellvmnpc}
    \BellVM\ is \NP-Complete.
\end{cor}
\begin{proof}
    \Cref{thm:EDPDT2BellVM} states that there exists a Karp reduction from 4-regular \EDPDT\ to \BellVM.
    This implies that \BellVM\ is \NP-hard, by \cref{cor:edpnpc}.
    Since any instance of \BellVM\ is also an instance of the general vertex-minor problem, which is in \NP~\cite{dahlberg2018long}, also \BellVM\ is in \NP and hence \NP-Complete.
\end{proof}

As a directly corollary we then also have that.

\begin{cor}
    \BellQM\ is \NP-Complete.
\end{cor}
\begin{proof}
	Follows directly from \cref{cor:bellvmnpc} and \cref{thm:QMVM}.
\end{proof}

The remaining part of this section will be used to prove \cref{thm:EDPDT2BellVM}.
The main part of the proof consists of proving \cref{lem:bellEDP_equiv} which provides an explicit reduction from $4$-regular \EDPDT\ to \BellVM.
This is done by constructing a $4$-regular multi-graph $H_{(G,T)}$ from an instance $(G,T)$ of $4$-regular \EDPDT, together with the graph $G_B=K_2^{\times\abs{T}}$ on a certain subset $B$ of the vertices of $H_{(G,T)}$.
The construction is such that $G_B$ is a vertex-minor of any alternance graph $\mathcal{A}(U)$ induced by an Eulerian tour $U$ on $H_{(G,T)}$ if and only if $(G,T)$ is a YES-instance of $4$-regular \EDPDT.
In other words, $(\mathcal{A}(U),B)$ is a YES-instance of \BellVM\ if and only if $(G,T)$ is a YES-instance of $4$-regular \EDPDT.
What is left to show is that reduction can be done in polynomial time, which is done in the proof of \cref{thm:EDPDT2BellVM} below.

\begin{thm}\label{thm:EDPDT2BellVM}
    The 4-regular \EDPDT\ problem is polynomially reducible to \BellVM.
\end{thm}
\begin{proof}
    From \cref{lem:bellEDP_equiv} below we see that any \emph{yes}(\emph{no})-instance of 4-regular \EDPDT\ can be mapped to a \emph{yes}(\emph{no})-instance of \BellVM.
    What remains to be shown is that this mapping can be performed in polynomial time.
    The reduction consist of the following to three steps:
    \begin{enumerate}
        \item Construct the multi-graph $H_{(G,T)}$ as defined in \cref{lem:bellEDP_equiv}.
        \item Find an Eulerian tour $U$ on $H_{(G,T)}$.
        \item Construct the alternance graph $\mathcal{A}(U)$ induced by $U$.
    \end{enumerate}

    Computing the graph $H_{(G,T)}$ can be done in polynomial time by simply adding the vertices and edges described in \cref{lem:bellEDP_equiv}.
    Note that the number of vertices and vertices in $H_{(G,T)}$ is $\abs{V(G)} + 2\cdot\abs{T}$ and $\abs{E(G)} + 4\cdot\abs{T}$.
    Finding an Eulerian tour $U$ on $H_{(G,T)}$ can be done in polynomial time~\cite{Fleury1883} in the size of $H_{(G,T)}$.
    Furthermore, constructing the alternance graph $\mathcal{A}(U)$ can be done in polynomial time as shown in~\cite{Bouchet1994circle}.
\end{proof}

\begin{lem}\label{lem:bellEDP_equiv}
	Let $(G,T=\{\{t_1,t'_1\},\dots,\{t_k,t_k'\}\})$ be an instance of 4-regular \EDPDT (see \cref{pbm:4_regEDPDTequiv}).
    As illustrated in \cref{fig:H_GTa}, let $H_{(G,T)}$ be the multi-graph constructed from $G$ by adding the distinct vertices $V_B=\bigcup_{i\in[k]}\{p_i,p_i'\}$ and the edges $(t_i,p_i)$, $(t_i',p_i')$, $(p_i,p_i')$ and\footnote{Note that $(p_i,p_i')$ is added twice.} $(p_i,p_i')$ for $i\in[k]$ and the edges $(p_i',p_{i+1})$ for $i\in[k-1]$ and the edge $(p_k',p_1)$.
    Let $G_B$ be the graph with vertices $V_B$ and edges $\bigcup_{i\in[k]}\{(p_i,p_i')\}$.
    Let $\mathcal{A}(U)$ be a circle graph described by the Eulerian tour $U$ on $H_{(G,T)}$.
    Then $G_B$ is a vertex-minor of $\mathcal{A}(U)$ if and only if $(G,T)$ is a \emph{yes}-instance of 4-regular \EDPDT.
\end{lem}

%

\begin{figure}[H]
    \centering
    \begin{subfigure}{0.9\textwidth}
        \includegraphics[width=1\textwidth]{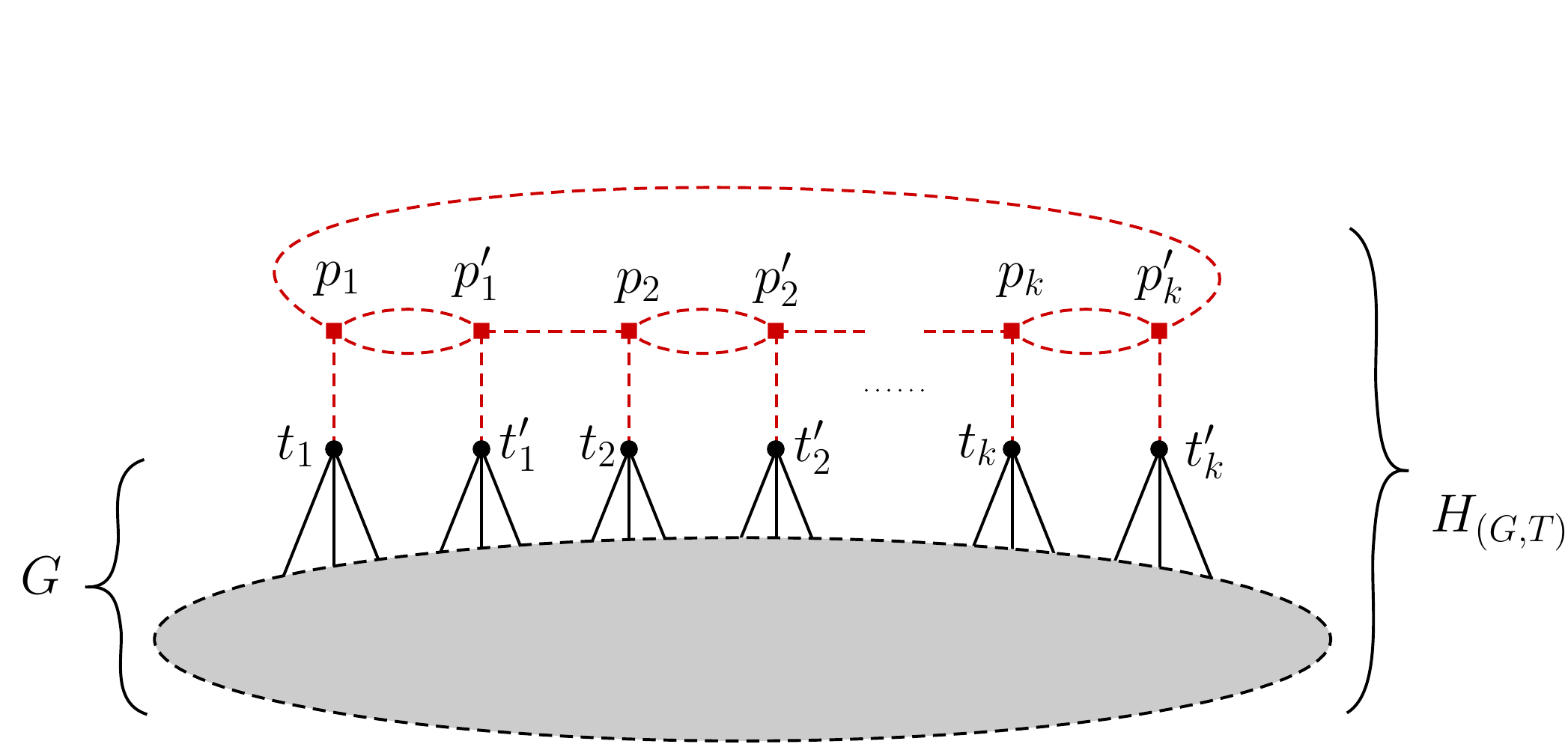}
        \caption{}
        \label{fig:H_GTa}
    \end{subfigure}
    \\
    \begin{subfigure}{0.9\textwidth}
        \centering
        \includegraphics[width=1\textwidth]{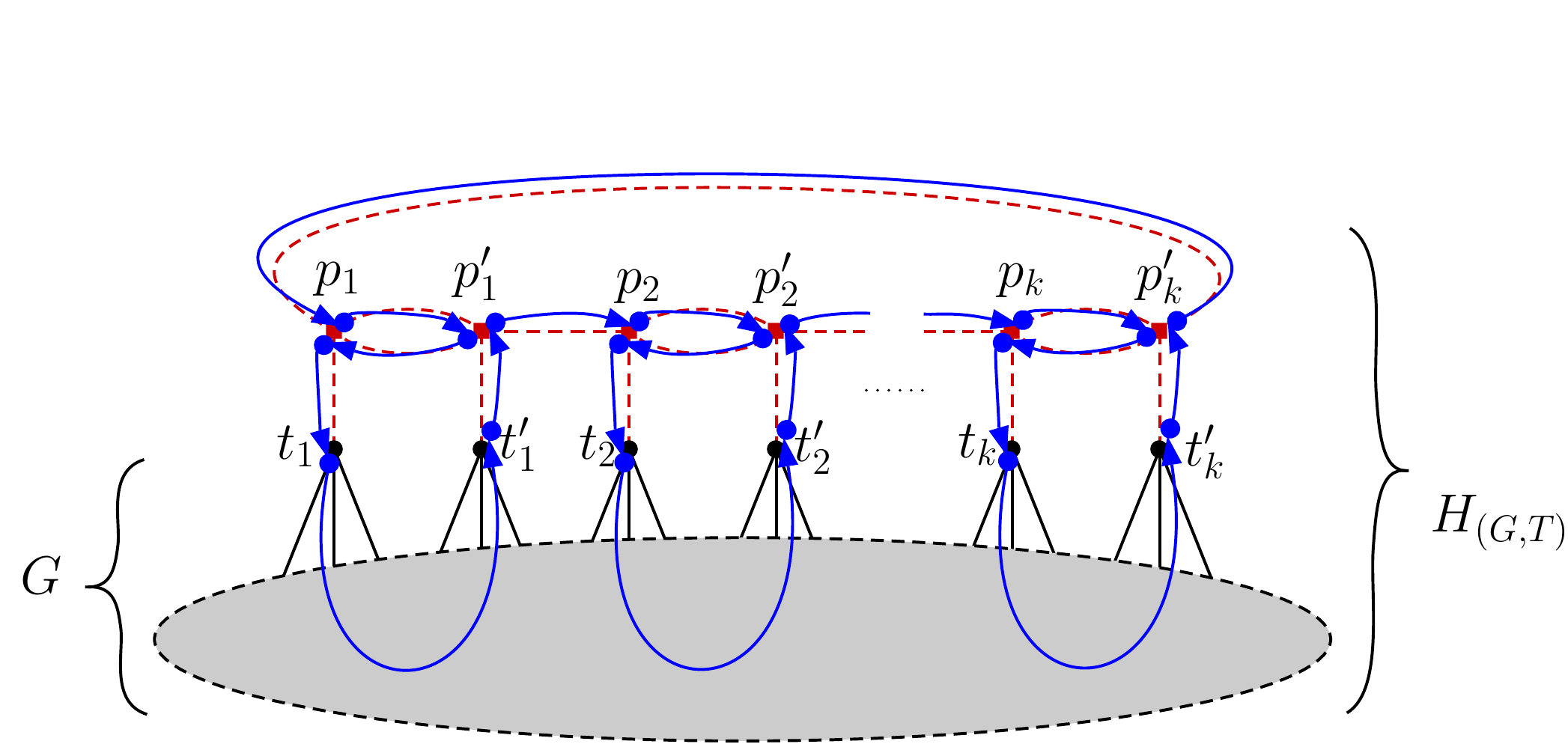}
        \caption{}
        \label{fig:H_GTb}
    \end{subfigure}
    \caption{\Cref{fig:H_GTa} shows the graph construction used to reduce 4-regular \EDPDT\ to \BellVM.
	The graph $G$ is a graph with only vertices of degree 4 except the vertices $\{t_1,t'_1,t_2,t'_2,\dots,t_k,t'_k\}$ which have degree 3.
	The vertices of $G$ not in $\{t_1,t'_1,t_2,t'_2,\dots,t_k,t'_k\}$ are visualized as a grey solid area in order to show that we put no further restrictions on $G$.
        $H_{(G,T)}$ is constructed from $G$ by adding the red square vertices and red dashed edges.
        In \cref{fig:H_GTb} the tour $U_0$ described by the word in \cref{eq:tour} is illustrated using solid blue arrows.
    }
    \label{fig:H_GT}
\end{figure}

\begin{proof}
    Let's first assume that $(G,T)$ is a YES-instance of 4-regular \EDPDT.
    This implies that there exist $k$ edge-disjoint paths $P_i$ for $i\in[k]$ such that the ends of $P_i$ are $t_i$ and $t_i'$.
    These paths also exist in $H_{(G,T)}$, since $G$ is a subgraph of $H_{(G,T)}$.
    Consider then the tour $U'$, see \cref{fig:H_GTb}, described by the word
    \begin{equation}\label{eq:tour}
        m(U')=p_1p_1'p_1m(P_1)p_1'p_2\dots p_{k-1}'p_kp_k'p_km(P_k)p_k'.
    \end{equation}
    Note $U'$ is not necessarily an Eulerian tour but can be extended to one using the efficient Hierholzer's algorithm~\cite{even_2011}.
    Denote the Eulerian tour obtained from $U'$ by $U$.
    Consider now the induced subgraph $\mathcal{A}(U)[V_B]$.
    Using \cref{eq:subword} we know that this induced subgraph is the same as the alternance graph given by the induced word 
    \begin{equation}
        m(U)[V_B]=p_1p_1'p_1p_1'p_2\dots p_{k-1}'p_kp_k'p_kp_k'.
    \end{equation}
    We then have that $\mathcal{A}(m(U)[V_B])$, and therefore $\mathcal{A}(U)[V_B]$, is the graph $G_B$.
    To see this, note that the only alternances of $m(U)[V_B]$ are $p_ip'_ip_ip'_i$ for $i\in[k]$ and the edges of $\mathcal{A}(m(U)[V_B])$ are therefore $(p_i, p'_i)$ for $i\in[k]$.
    Since $\mathcal{A}(U)[V_B]$ is an induced subgraph of $\mathcal{A}(U)$ it is by definition also a vertex-minor of $\mathcal{A}(U)$.
    As we saw above, $G_B$ is equal to $\mathcal{A}(U)[V_B]$ and is therefore also a vertex-minor of $\mathcal{A}(U)$.

    Let's now instead assume that $(G,T)$ is a NO-instance of 4-regular \EDPDT.
    We will show that $G_B$ is not a vertex-minor of $\mathcal{A}(U)$.
    Let $U$ be an Eulerian tour on $H_{(G,T)}$, which exist since $H_{(G,T)}$ is 4-regular.
    Consider the sub-trails of $U$ for which both ends are in $V_B$ and that use no edges between vertices in $V_B$.
    Note that there are $k$ such sub-trails since each vertex in $V_B$ only has one edge which does not go to another vertex in $V_B$.
    Let's denote these sub-trails as $\tilde{P}_i$ for $i\in[k]$ and their ends as $\{p_i,\tilde{p}_i'\}$.
    From the assumption that $(G,T)$ is a \emph{no}-instance of 4-regular \EDPDT, we know that $\{p_i,\tilde{p}_i'\}\neq\{p_i,p_i'\}$ for at least one $i\in[k]$.
    Consider now the graph $\mathcal{A}(U)[V_B]$.
    We will now show that $\mathcal{A}(U)[V_B]$ cannot be $G_P$.
    Since $U$ was an arbitrary Eulerian tour on $H_{(G,T)}$, we then know that $G_B$ is not a vertex-minor of $U$, by \cref{thm:vm_of_eul}.

    To show that $\mathcal{A}(U)[V_B]$ is not $G_B$, consider the induced subgraph $H_{(G,T)}[V_B]$.
    Note that all vertices in $H_{(G,T)}[V_B]$ have degree 3, see \cref{fig:H_GTa}.
    Let $H_{(U,B)}$ be the graph obtained from $H_{(G,T)}[V_B]$ by adding the edges $(p_i,\tilde{p}_i')$.
    Note now that $m(U)[V_B]$ also describes an Eulerian tour on $H_{(U,B)}$.
    Thus, $\mathcal{A}(U)[V_B]$ is the alternance graph $\mathcal{A}(U')$ for some Eulerian tour $U'$ on $H_{(U,B)}$.
    Therefore, by using \cref{lem:notBells} we know that $\mathcal{A}(U)[V_B]$ is not $G_B$.
\end{proof}


\begin{figure}[H]
    \centering
    \includegraphics[width=0.3\textwidth]{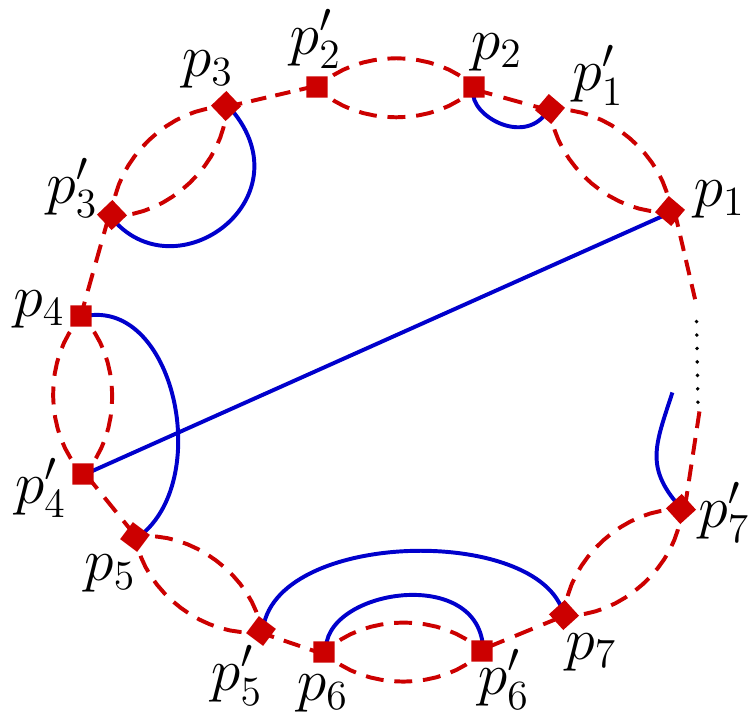}
    \caption{
	    The graph $H_{(B,\tilde{E}_B)}$ used in \cref{lem:notBells}, where $\tilde{E}_B$ are the solid blue edges.
        \Cref{lem:notBells} states that there exist an Eulerian tour $U$ on $H_{(B,\tilde{E}_B)}$ such that $\mathcal{A}(U)=G_B$ if and only if all the solid blue edges $\tilde{E}_B$ are between vertices $p_i$ and $p'_i$ for $i\in[k]$.}
    \label{fig:H_PP}
\end{figure}

\begin{lem}\label{lem:notBells}
    Let $B$ be the set $\{p_1,p_1',\dots,p_k,p_k'\}$.
    Let $C_B$ be the cycle graph on vertices $\bigcup_{i\in[k]}\{p_i,p_i'\}$ and with edges $\{(p_1,p_1'),(p_1',p_2),\dots (p_k,p_k'),(p_k',p_1)\}$.
    Let $\tilde{C}_B$ be the multi-graph obtained from $C_B$ by duplicating the edges $\{(p_i,p_i')\}_{i\in[k]}$.
    Let $\tilde{E}_B=\{(p_i,\tilde{p}_i')\}_{i\in[k]}$ be edges such that if they are added to $\tilde{C}_B$, the graph obtained is 4-regular.
    Denote the graph obtained from $\tilde{C}_B$ by adding the edges $\tilde{E}_B$ by $H_{(B,\tilde{E}_B)}$.
    Let $G_B$ be the graph $(B,\bigcup_{i\in[k]}\{(p_i,p_i')\}$.
    Then for any Eulerian tour $U$ on $H_{(B,\tilde{E}_B)}$, the graph $\mathcal{A}(U)$ is equal to $G_B$ if and only if $\{p_i,\tilde{p}_i'\}=\{p_i,p_i'\}$ for all $i\in[k]$.
\end{lem}
\begin{proof}
    First note that if $\mathcal{A}(U)=G_B$ for some Eulerian tour $U$ on $H_{(B,\tilde{E}_B)}$ then $\mathcal{A}(U')=G_B$ for any Eulerian tour $U'$ on $H_{(B,\tilde{E}_B)}$.
    This is because the graph $G_B$ is invariant under local complementations.

    Assume now that $\{p_i,\tilde{p}_i'\}=\{p_i,p_i'\}$ for all $i\in[k]$.
    Then the tour $U$ described by the double occurrence word
    \begin{equation}
        m(U)=p_1p_1'p_1p_1'p_2\dots p_{k-1}'p_kp_k'p_kp_k'.
    \end{equation}
    is an Eulerian tour on $H_{(B,\tilde{E}_B)}$.
    Furthermore, $\mathcal{A}(U)=G_B$.

    Assume now on the other hand that $\{p_i,\tilde{p}_i'\}\neq\{p_i,p_i'\}$ for some $i\in[k]$.
    Let $(p_i,\tilde{p}_i')$ be a pair for which $\{p_i,\tilde{p}_i'\}\neq\{p_i,p_i'\}$.
    Consider now the tour
    \begin{equation}
        U'=p_i(p_i,\tilde{p}_i')P_C
    \end{equation}
    on $H_{(B,\tilde{E}_B)}$ where $P_C$ is a path in $C_B$ which ends at $p_i$, i.e.
$E(P_C)\cap \tilde{E}_B=\emptyset$.
    By Hierholzer's algorithm let $U$ be an Eulerian tour on $H_{(B,\tilde{E}_B)}$ obtained by extending $U'$.
    Note now that $(p_i,\tilde{p}_i')$ is an alternance in $m(U)$ and therefore an edge in $\mathcal{A}(U)$.
    But since $(p_i,\tilde{p}_i')$ is not an edge in $G_B$, we know that $\mathcal{A}(U)$ is not equal to $G_B$.
\end{proof}

\section{Conclusion}\label{sec:conclusion}
The problem of transforming graph states to Bell-pairs using local operations is a problem with direct applications to the development of quantum networks or distributed quantum processors.
Solutions to special cases of this problem have been considered in for example~\cite{pant2019routing} and~\cite{vanmeter2018butterfly}.
However, at least to our knowledge, the computational complexity of this problem was previously unknown.
Here we show that deciding whether a given graph state $\ket{G}$ can be transformed into a set of Bell pairs on a given set of vertices, using only single-qubit Clifford operations, single-qubit Pauli measurements and classical communication is \NP-Complete.
In fact, we show that the problem remains \NP-Complete if $G$ is a circle graph.

\section*{Acknowledgements}
AD, JH and SW were supported by an ERC Starting grant, and NWO VIDI grant, and Zwaartekracht QSC.

\bibliographystyle{unsrtnat}
\bibliography{small}

\newpage

\appendix

\section{The 4-regular \EDPDT\ problem is NP-Complete}\label{sec:edpnpc}
From \cite{vygen1995np} we know that the EDP problem is \NP-Complete even when the graph $G\cup D$ is Eulerian.
Here we prove that this problem remains \NP-Complete if we restrict the demand graph to be of the form $D = K_2^{\times k}$ and restrict $G\cup D$ to be $4$-regular. We call this problem 4-regular \EDPDT, see \cref{pbm:4_regEDPDTequiv}.
To do so we will first introduce the notion of a grid graph gadget, an essential tool for reducing the Eulerian EDP problem to the $4$-regular EDP problem.
We make use of these results to show \BellVM\ (and \BellQM) is \NP-Complete in \cref{sec:bellvmnpc}.

\begin{mydef}\label{def:grid_gadget}
    Let $G$ be an Eulerian multi-graph and let $v$ be a vertex of $G$ of degree $2n$ with incident edges labeled $1,\ldots n, 1', \ldots, n'$.
    The grid gadget gadget $GG_v$ associated to $v$, illustrated in \cref{fig:grid_gadget}, is a graph on $n^2+2n$ vertices labeled
\begin{equation}
    V(GG_v) = \Big\{v_{i} : i\in [n]\Big\}\cup \Big\{v_{i'} : i\in [n]\Big\}\cup \Big\{v_{i,j'} : i, j \in [n]\Big\}\setminus\{v_{0,0'}\}
\end{equation}
with edge set
\begin{align}
    E(GG_v) &= \Big\{ (v_i,v_{i+1} ) : i \in [n-1]\Big\}\cup \{(v_n,v_1)\}\nonumber\\
    &\hspace{5em} \cup \Big\{ (v_{i'},v_{(i+1)'} ) : i \in [n-1]\Big\}\cup \{(v_{n'},v_{1'})\}\nonumber\\
&\hspace{5em}\cup \left[\bigcup_{i\in [n]} \Big\{(v_i,v_{i,1'})\Big\}\cup \Big\{(v_{i,j'},v_{i,(j+1)'}) : j \in [n-1]\Big\}\right]\nonumber\\
&\hspace{5em}\cup \left[\bigcup_{i\in [n]} \Big\{(v_{i'},v_{1,i'})\Big\}\cup \Big\{(v_{j,i'},v_{j+1,i'}) : j \in [n-1]\Big\}\right]\nonumber\\
&\hspace{5em}\cup  \Big\{(v_{i,n'},v_{n,i'}) : i \in [n]\Big\}\cup \{(v_1,v_{1'})\,.
\end{align}
\end{mydef}

For later convenience we also define the following subsets of edges for $i \in [n]$
\begin{align}
    E^{\rm hor}_i[j',k'] &=\Big\{(v_{i,l'},v_{i,(l+1)'}) : l \in [j-1,k-1]\Big\} \quad \mathrm{for} \quad i,j,k\in[n]\label{eq:horver1}\\
    E^{\rm ver}_{i'}[j,k] &= \Big\{(v_{l,i'},v_{l+1,i'}) : l \in [j-1,k-1]\Big\} \quad \mathrm{for} \quad i,j,k\in[n].\label{eq:horver2}
\end{align}
where $v_{i,0'}=v_i$ and $v_{0,i'}=v_{i'}$.
Informally, with respect to \cref{fig:grid_gadget} the edges in $E^{\rm hor}_i[j',k']$ are the horizontal edges at level $i$ which are right of $j'$ and left of $k'$.
Similarly, the edges in $E^{\rm ver}_{i'}[j,k]$ are the vertical edges at level $i'$ which are above $j$ below $k$.

\begin{figure}[H]
    \centering
    \includegraphics[width=0.8\textwidth]{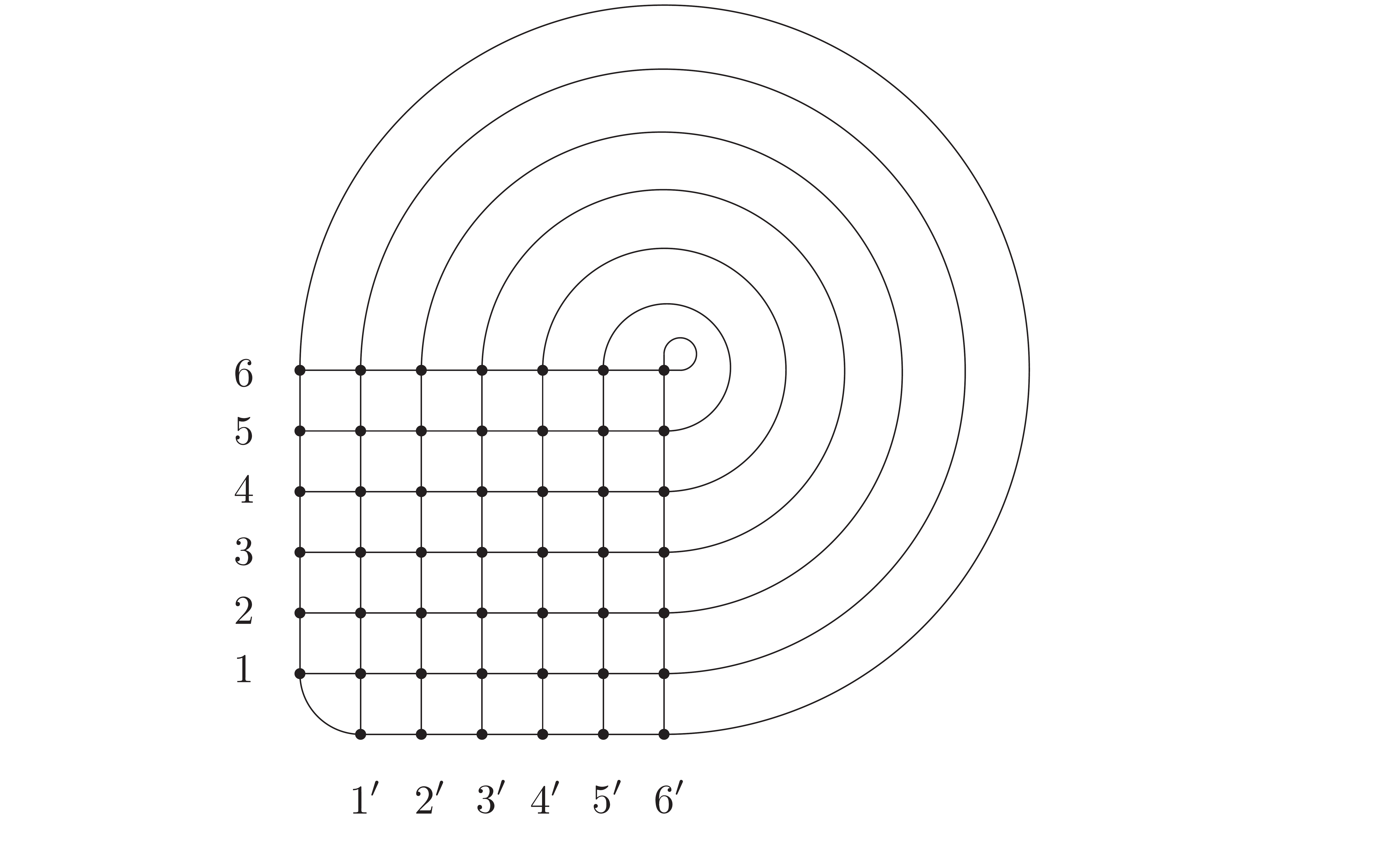}
    \caption{An example of the grid gadget defined in \cref{def:grid_gadget}.}
    \label{fig:grid_gadget}
\end{figure}

Note that all the vertices in the grid graph gadget $GG_v$ have degree $4$ or $3$, with only the vertices on the 'outside' (associated to the incident edges of the vertex $v$) having degree $3$.


We will now prove that the grid gadget, in a specified sense, provides all-to-all connectivity.

\begin{lem}\label{lem:alltoall}
Let $v$ be a vertex in a graph with $2n$ incident edges labeled $1,\ldots, n, 1', \ldots n'$ and let $GG_v$ be the associated grid graph gadget as defined in \cref{def:grid_gadget}.
Consider an arbitrary pairing of the vertices $v_1,\ldots, v_n,v_{1'},\ldots,v_{n'}$ consisting of tuples of the form $(v_i,v_j)$ (unprimed pairs),  $(v_k,v_{l'})$ (mixed pairs) and $(v_{m'},v_{n'})$ (primed pairs).
Then there exist edge-disjoint paths $P_1,\ldots P_n$ that connect all pairs and moreover contain all edges in $GG_v$.
\end{lem}
\begin{proof}
Assume without of loss of generality that $i<j$ for the unprimed pairs and $m'<n'$ for the primed pairs (i.e. the first index is always smaller).
We will prove the lemma by explicit construction of the paths $P_1,\ldots, P_n$
Let $M$ be the list of mixed pairs, $L$ be the list of unprimed pairs and $L'$ be the list of primed pairs.
Note that $L$ and $L'$ are necessarily of equal length.
Consider first the list of mixed pairs $M$.
Construct the paths $P_{k,l'}$ by the following algorithm.
An example of these paths is shown in \cref{fig:grid_gadget_example}.

\begin{algorithm}
    \caption{Algorithm to construct the paths $P_{k,l'}$.}
    \label{alg:Ps_mixed}
    \begin{algorithmic}
        \For {$x \in [K]$ where $K$ is the length of $M$}
            \State Create the path $P_{k_x,l_x'}$ associated to the $x$'th pair in $M$ by:
            \State\hspace{3em} (1) Walking horizontally rightwards from $v_{k_x}$ to $v_{k_x,l'_x}$
            \State\hspace{3em} (2) Walking vertically downwards from $v_{k_x,l'_x}$ to  $v_{l'_x}$
        \EndFor
    \end{algorithmic}
\end{algorithm}

Note that for every pair we have that
\begin{equation}
	E(P_{k,l'}) = E^{\rm hor}_{k}[0,l']\cup E^{\rm ver}_{l'}[0,k]
\end{equation}
where $E^{\rm hor}_{k}[0,l']$ and $E^{\rm ver}_{l'}[0,k]$ are defined in \cref{eq:horver1,eq:horver2}.
Note also that since these sets are all disjoint the paths $P_{k,l'}$ are all mutually edge-disjoint.\\

Now consider the lists of of primed and unprimed pairs $L$ and $L'$.
We will construct paths $P_{i,j}$ (unprimed) and $P_{m',n'}$ (primed) using the following algorithm.


\begin{algorithm}
    \caption{Algorithm to construct the paths $P_{i,j}$ and $P_{i',j'}$.}
    \label{alg:Ps}
    \begin{algorithmic}
	    \For {$x \in [K]$ where $K$ is the length of $L$} \Comment{$L$ and $L'$ are necessarily of the same length.}
            \State Create the path $P_{i_x,j_x}$ associated to the $k$'th pair in $L$ by:
	    \State\hspace{3em} (1) Walking horizontally rightwards from $v_{i_x}$ to $v_{i_x,m'_x}$
	    \State\hspace{3em} (2) Walking vertically upwards from $v_{i_x,m'_x}$ to  $v_{j_x,m'_x}$
	    \State\hspace{3em} (3) Walking horizontally leftwards from $v_{j_x,m'_x}$ to $v_{j_x}$
            \State Create the path $P_{m_x',n_x'}$ associated to the $x$'th pair in $L'$ by:
	    \State\hspace{3em} (1) Walking vertically upwards from $v_{m'_x}$ to $v_{i_x,m'_x}$
	    \State\hspace{3em} (2) Walking horizontally rightwards from $v_{i_x,m'_x}$ to  $v_{i_x,n'_x}$
	    \State\hspace{3em} (3) Walking vertically downwards from $v_{i_x,n'_x}$ to $v_{n'_x}$
        \EndFor
    \end{algorithmic}
\end{algorithm}

\begin{figure}[H]
    \centering
    \includegraphics[width=0.8\textwidth]{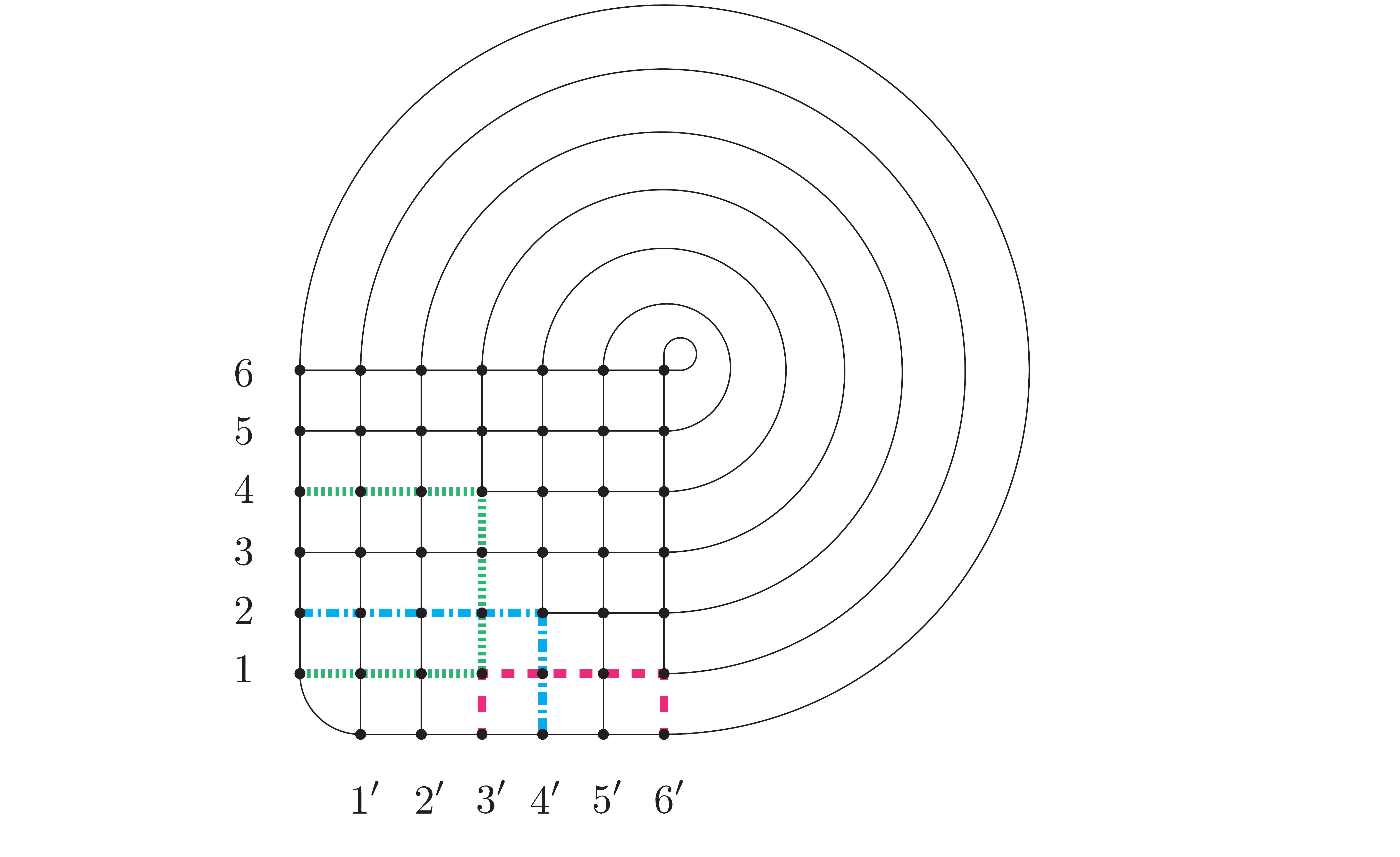}
    \caption{An example of paths produced by \cref{alg:Ps_mixed,alg:Ps} on the grid gadget defined in \cref{def:grid_gadget}.
    In this example one of the mixed pairs in $M$ is $(2,4')$ giving rise to the blue dashed-dotted path.
    Furthermore, the first unprimed pair in $L$ is $(1,4)$ and the first primed pair in $L'$ is $(3',6')$, giving rise to the green dotted path and the red dashed path respectively, meeting at $v_{2,3'}$.
    Note that the $x$'th primed and unprimed paths always meet at the vertex $v_{i_x,m'_x}$ where $v_{i_x}$ and $v_{m'_x}$ is the first vertex in the $x$'th primed and unprimed pair respectively.
    }
    \label{fig:grid_gadget_example}
\end{figure}

Now note that the for all $x \in [K]$ (where $K$ is the length of $L$) we have that 
\begin{align}
	E(P_{i_x,j_x}) &= E^{\rm hor}_{i_x}[0,m'_x]\cup E^{\rm ver}_{m'_x}[i_x,j_x]\cup E^{\rm hor}_{j_x}[0,m'_x]\label{eq:collis1}\\
	E(P_{m_x',n_x'}) &= E^{\rm ver}_{m'_x}[0,i_x]\cup E^{\rm hor}_{i_x}[m'_x,n'_x]\cup E^{\rm ver}_{n'_x}[0,i_x]\label{eq:collis2}
\end{align}
This immediately implies that all $P_{i_x,j_x}$ and $P_{m'_x,n'_x}$ are mutually edge-disjoin and moreover that no $P_{i_x,j_x}$ nor any $P_{m'_x,n'_x}$ share edges with any of the mixed-pair paths $P_{k,l'}$.
This last point can be seen by noting that an unprimed vertex $v_i$ can either be part of a mixed pair or an unprimed pair but not both at the same time (with a similar argument for the primed vertices).

Hence we have constructed a set of edge-disjoint paths that connect all vertex pairs.
However they do not yet contain all edges in $E(GG_v)$.
It is however straightforward to extend the paths to include all remaining edges.
To see this consider the grid graph gadget $GG_v$ and remove all edges that are contained in one of the paths constructed above.
What remains is a not necessarily connected graph $G_{\rm rem}$ of which the connected components, by construction, share a vertex with at least one of the constructed paths.
Moreover, all these graphs will be Eulerian.
Now choose for each graph $G_{\rm rem}$ a Eulerian tour $U$ and insert it into exactly one of the paths that shares a vertex with $G_{\rm rem}$.
The resulting set of paths will still have mutually edge disjoint elements (since a Eulerian tour is edge-disjoint by definition) and furthermore the union of all the paths in the set contains all edges in the grid graph gadget $GG_v$.
This completes the proof.
\end{proof}

We will now make use of \cref{lem:alltoall} to map instances of \EDP\ to instances of $4$-regular \EDPDT.
We will first construct a mapping from arbitrary demand graphs $D$ to demand graphs of the form $K_2^{\times k}$.
Then, to make the graph $G\cup D$ $4$-regular, we replace any higher-degree vertices in $G$ with the grid gadget above.
Using \cref{lem:alltoall}, we can prove that this puts no restriction on the possible paths.

\begin{thm}\label{thm:EDP2EDPDT}
Let $G$ and $D$ be graphs such that $V(D)\subseteq V(G)$ and $G\cup D$ is Eulerian.
There exist graphs $G''$ and $D'$ such that $D' = K_2^{\times k}$ where $k =\abs{E(D)}$, $G''\cup D'$ is $4$-regular and $(G,D)$ is a YES-instance of the EDP problem if and only if $(G'', D')$ is.

\end{thm}
\begin{figure}[H]
    \centering
    \includegraphics[width=0.6\textwidth]{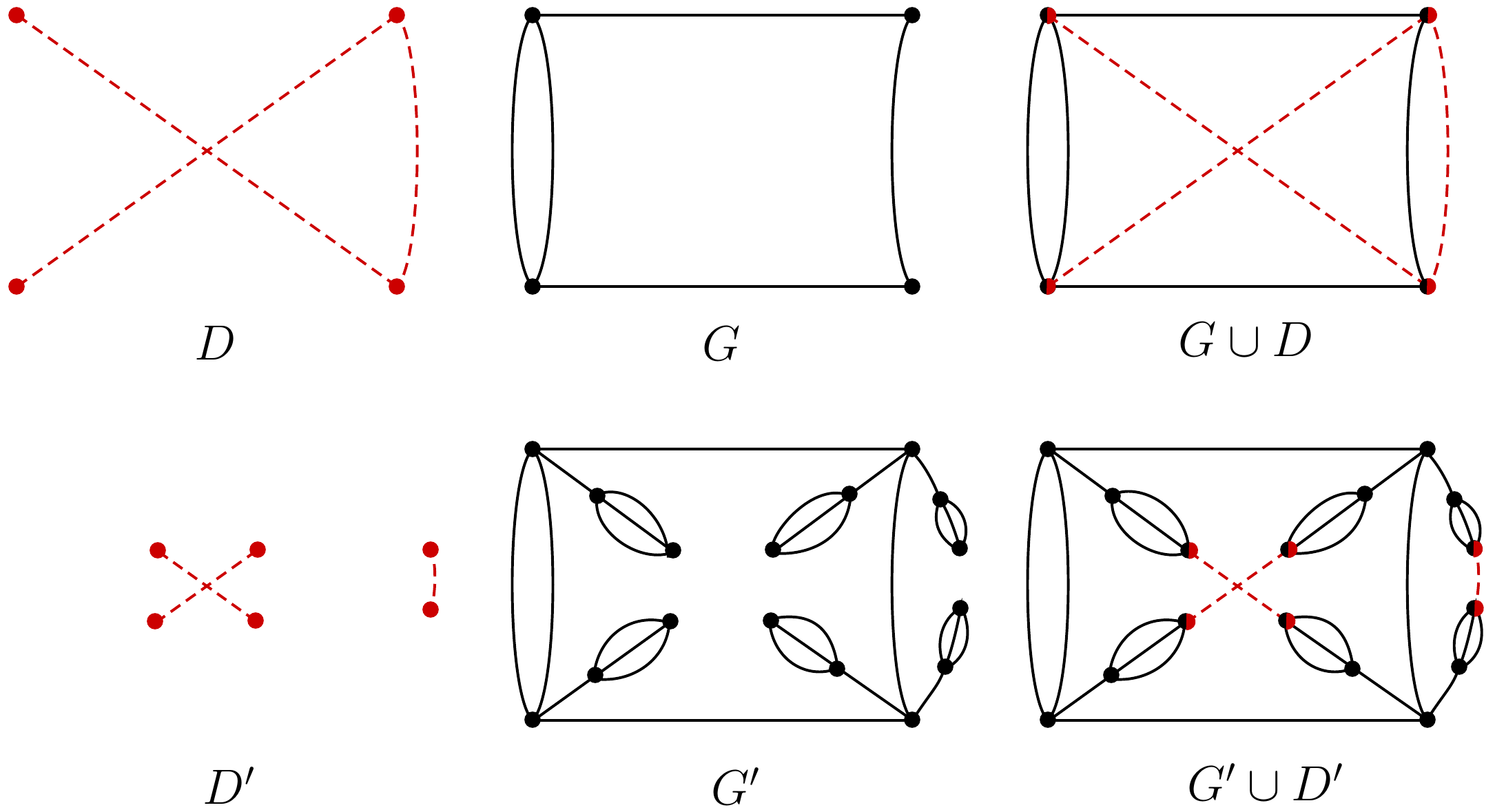}
    \caption{An example showing the mapping from an instance $(G,D)$ of EDP to an instance $(G',D')$ of EDPDT.
    Black solid edges are edges from $G$ or $G'$ and red dashed edges are edges from $D'$.}
    \label{fig:grape_expansion}
\end{figure}

\begin{proof}
    To prove the theorem we will construct an explicit mapping from the graphs $(G,D)$ to $(G'',D')$.
    We do this in two steps: (1) map $(G,D)$ to $(G',D')$ where $D'=K_2^{\times k}$ but $G'\cup D'$ is not necessarily 4-regular and (2) map $(G',D')$ to $(G'',D')$  such that $G''\cup D'$ is 4-regular.
    The first mapping is visualized in \cref{fig:grape_expansion} and formalized in \cref{eq:VGp,eq:EGp,eq:VDp,eq:EDp}.
    Informally, we replace each edge $e=(u,v)\in E(D)$ in the graph $G\cup D$ with the following gadget
    \begin{equation}\label{eq:grape_expandsion_gadget}
        \raisebox{-0.03\textwidth}{\includegraphics[scale=0.8]{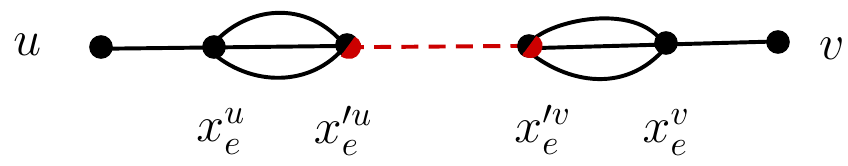}}.
    \end{equation}
    where $x^u_e$, $x'^u_e$, $x'^v_e$ and $x^v_e$ are all new vertices.
    Note that we label the vertex closer to the corresponding vertex $u$ in the edge $e$ without the prime and the further one with a prime.
    Formally, we define $G'$ to be
    \begin{align}
	    V(G') &= V(G)\cup\Bigg( \bigcup_{e=(u,v)\in E(D)} \{x^u_e, x'^u_e, x'^v_e, x^v_e\} \Bigg) \label{eq:VGp}\\
	    E(G') &= E(G)\cup\Bigg( \bigcup_{e=(u,v)\in E(D)} \Big\{(u,x^u_e),(x^u_e,x'^u_e),(x^u_e,x'^u_e),(x^u_e,x'^u_e), \nonumber\\
	      &\hspace{12em}(x'^v_e,x^v_e),(x'^v_e,x^v_e),(x'^v_e,x^v_e),(x^v_e,v)\Big\} \Bigg).\label{eq:EGp}
    \end{align}
    Note that we have added the edges $(x^u_e,x'^u_e)$ and $(x'^v_e,x^v_e)$ three times.
    We now define the new demand graph $D'$ as
        \begin{align}
        V(D') &= \bigcup_{e=(u,v)\in E(D)} \{x'^u_e, x'^v_e\} \label{eq:VDp}\\
	E(D') &= \bigcup_{e=(u,v)\in E(D)} \{(x'^u_e, x'^v_e)\} \label{eq:EDp}
        \end{align}
    Note that $D' = K_2^{\times k}$.
    Note also that $G'\cup D'$ is still Eulerian.
    However it is in general not $4$-regular.\\

    Now let $N$ be the set of vertices of $G'\cup D'$ of degree different than $4$.
    Note that, by construction $N \cap V(D') = \emptyset$.
    That is, all vertices of degree other that $4$ are exclusively vertices in $G'$.
    We will from $G'$ construct a graph $G''$ such that $G'' \cup D'$ is $4$-regular.
    For every vertex in $N$ which has degree 2, we simply add a self-loop to the vertex, making the vertex have degree 4 without changing any connectivity.
    Furthermore, for every vertex $v\in N$ with degree larger than 4, we replace it by the grid graph gadget $GG_v$ as defined in \cref{def:grid_gadget}, attaching the edges incident on $v$ to the vertices $v_1, \ldots v_n,v_{1'}, \ldots,v_{n'}$ (where $n={\rm deg}(v)/2$) in the grid graph gadget.
    Note the graph $G''\cup D'$ obtained after this procedure is $4$-regular.


    To prove the theorem we now need to show that $(G,D)$ is a YES-instance to EDP if and only if $(G'',D')$ is a YES-instance.
    Again, we will do this in two steps: (1) first show that $(G,D)$ is a YES-instance if and only if $(G',D'$) is and (2) show that $(G',D')$ is a YES-instance if and only if $(G'',D')$ is.
    \begin{enumerate}
	    \item

    We now prove that $(G',D') $ is a YES-instance of EDP if and only $(G,D)$ is.
    First assume $(G,D)$ is a YES-instance of EDP.
    This means there exists for each $e = (u,v)\in E(D)$ a path $P_e$ on $G$ that begins and ends at the vertices $u,v$ such that all paths $P_e$ are mutually edge-disjoint.
    Now, for each $e\in E(D)$ define the path $P_e'$ on $G'$ as
        \begin{equation}\label{eq:edge_disjoint_path}
	P_e' = x'^u_ex^u_ex'^u_ex^u_euP_evx^v_ex'^v_ex^v_ex'^v_e
        \end{equation}
    where we have omitted writing out the edges that the path traverses for clarity.
    For a visual aid refer to \cref{eq:grape_expandsion_gadget} where we instead of starting the path at $u$ as in $P_e$ we start at $x'^u_e$, traverse the edge $(x^u_x,x'^u_e)$ back and fourth three times and then move to $u$ using the edge $(x^u_e,x'^u_e)$.
    From $u$ the path $P_e'$ is the same as $P_e$ and when arriving at $v$ we instead end at $x'^v_e$ similarly to how we started at $x'^u_e$.
    Thus $P'_e$ is a path connecting the vertices $x'^u_e,x'^u_e$.
    Note that by definition $(x'^u_e,x'^u_e)$ is an edge in the demand graph $D'$ (precisely corresponding to the edge $(u,v)\in D$).
    The paths $P_e'$  are also mutually edge-disjoint, since the paths $P_e$ are.
    Hence $(G',D')$ is a YES-instance.

    For the other direction, assume that $(G',D')$ is a YES-instance of EDP and thus that there exists edge-disjoint paths $P_e'$ connecting the vertices $x'^u_e,x'^v_e$ for all $e=(x'^u_e,x'^v_e)\in E(D')$.
    One can then see that $P'_e$ is either of the form as in \cref{eq:edge_disjoint_path} or as
	\begin{equation}
	P_e' = x'^u_ex^u_euP_evx^v_ex'^v_e.
	\end{equation}
	This means that the associated path $P_e$ forms a path between $u$ and $v$ in $G$ for all $e=\in E(D')$.
	Furthermore, since all $P'_e$ are pairwise edge-disjoint, so are the $P_e$.
    Hence $(G,D)$ is also a YES-instance of EDP.\\
    
    \item

    Next we argue that $(G',D')$ is a YES-instance of EDP if and only if $(G'',D')$ is.
If $(G'',D')$ is a YES-instance of EDP then so is $(G',D')$, since any edge-disjoint paths passing through a grid graph gadget $GG_v$ can also be made into edge-disjoint paths passing through the vertex $v$.
Hence assume that $(G',D')$ is a YES-instance of EDP.
    Note that the only difference between $(G',D')$ and $(G'',D')$ is the replacement of vertices $v\in N$ with the grid graph gadget $GG_v$.
    Moreover recall that $N\cap V(D')=\emptyset$.
Finally, by \cref{lem:alltoall}, we know that any ${\rm deg}(v)/2$ paths passing through a vertex $v$ can be mapped to ${\rm deg}(v)/2$ paths passing through $GG_v$, and that these paths are mutually edge-disjoint and use all edges in $GG_v$.
This implies that if $(G',D')$ is a  YES-instance of EDP, then so is $(G'',D')$.
This proves the theorem.
    \end{enumerate}
\end{proof}

We can now prove \cref{cor:edpnpc}.
\begin{proof}[Proof of \cref{cor:edpnpc}]
    \Cref{thm:EDP2EDPDT} states that there exists a many-one reduction from Eulerian \EDP\ to $4$-regular \EDPDT.
    Furthermore, this reduction consists of constructing the graphs $(G'',D')$ from the graphs $(G,D)$ by the explicit rules in \cref{eq:VGp,eq:EGp,eq:VDp,eq:EDp} and by replacing vertices of degree more than $4$ with the grid graph gadget, which is clearly polynomial.
    This shows that $4$-regular \EDPDT\ is \NP-Hard.
    Since any instance of $4$-regular \EDPDT\ is also an instance of the general \EDP\, which is in \NP, also $4$-regular \EDPDT\ is in \NP and thus \NP-Complete.
\end{proof}

\end{document}